\documentclass[a4paper, 11pt]{article}
\usepackage{amsmath}
\usepackage{amssymb}
\usepackage{amsthm}
\usepackage{physics}
\usepackage{mathtools}
\usepackage{thmtools} 
\usepackage{thm-restate}
\usepackage{bbm}
\usepackage{cancel}

\allowdisplaybreaks

\PassOptionsToPackage{hyphens}{url}
\usepackage{url}
\usepackage{hyperref}
\hypersetup{colorlinks=true, linkcolor=Black, citecolor=Blue, urlcolor=Blue,hypertexnames=false}
\usepackage[capitalize]{cleveref}
\usepackage{bookmark}

\usepackage{tabularx}
\usepackage{booktabs}
\usepackage{multicol} 
\usepackage{makecell}
\usepackage{array,multirow}
\usepackage{caption}
\usepackage{subcaption}
\usepackage{graphicx}

\usepackage[dvipsnames]{xcolor}
\usepackage{csquotes}
\usepackage{datetime} 
\usepackage{datetime2}
\usepackage{fancyhdr}
\usepackage{titlesec}
\usepackage{sectsty} 
\usepackage{listings}
\usepackage{enumitem}
\usepackage{authblk}
\usepackage{svg}
\usepackage[left=1.15in, right=1.15in, top=1.2in, bottom=1in, headsep=.25in]{geometry}
\usepackage[bottom]{footmisc}
\usepackage[english]{babel}
\usepackage{tikz}
\usepackage{tikz-cd}
\usepackage{pgfplots}
\pgfplotsset{compat=1.18}
\usepgfplotslibrary{fillbetween}
\usepackage{relsize}
\usetikzlibrary{decorations.fractals,positioning}
\usetikzlibrary{arrows.meta}
\usetikzlibrary{positioning, shapes.geometric, calc}
\usetikzlibrary{shadows.blur}

\usetikzlibrary{decorations.pathmorphing}
\usetikzlibrary{arrows.meta}
\usetikzlibrary{decorations.markings}

\allowdisplaybreaks

\sectionfont{\fontsize{13}{15}\selectfont}
\subsectionfont{\fontsize{11}{15}\selectfont}

\renewenvironment{abstract}
{\small\begin{quote}\noindent \par{\sc \abstractname.}}
{\noindent\end{quote}}

\usepackage[sorting=none,backend=biber,style=numeric-comp,maxbibnames=99,maxcitenames=99]{biblatex}
\DeclareFieldFormat{titlecase}{\MakeTitleCase{#1}}

\newrobustcmd{\MakeTitleCase}[1]{%
  \ifthenelse{\ifcurrentfield{booktitle}\OR\ifcurrentfield{booksubtitle}%
    \OR\ifcurrentfield{maintitle}\OR\ifcurrentfield{mainsubtitle}%
    \OR\ifcurrentfield{journaltitle}\OR\ifcurrentfield{journalsubtitle}%
    \OR\ifcurrentfield{issuetitle}\OR\ifcurrentfield{issuesubtitle}%
    \OR\ifentrytype{book}\OR\ifentrytype{mvbook}\OR\ifentrytype{bookinbook}%
    \OR\ifentrytype{booklet}\OR\ifentrytype{suppbook}%
    \OR\ifentrytype{collection}\OR\ifentrytype{mvcollection}%
    \OR\ifentrytype{suppcollection}\OR\ifentrytype{manual}%
    \OR\ifentrytype{periodical}\OR\ifentrytype{suppperiodical}%
    \OR\ifentrytype{proceedings}\OR\ifentrytype{mvproceedings}%
    \OR\ifentrytype{reference}\OR\ifentrytype{mvreference}%
    \OR\ifentrytype{report}\OR\ifentrytype{thesis}}
    {#1}
    {\MakeSentenceCase{#1}}}
\usepackage[subfigure]{tocloft}
\newtheorem{theorem}{Theorem}
\newtheorem{definition}[theorem]{Definition}
\newtheorem{lemma}[theorem]{Lemma}
\newtheorem{remark}[theorem]{Remark}
\newtheorem*{example*}{Example}
\newtheorem{proposition}[theorem]{Proposition}
\newtheorem{corollary}[theorem]{Corollary}

\newtheoremstyle{named}{}{}{\itshape}{}{\bfseries}{.}{.5em}{\thmnote{#3}}
\theoremstyle{named}

\makeatletter
\def\namedlabel#1#2{\begingroup
   \def\@currentlabel{#2}%
   \label{#1}\endgroup
}
\makeatother

\crefname{assumption}{assumption}{assumptions}

\crefname{equation}{}{}
\Crefname{equation}{}{}

\makeatletter
\AddToHook{cmd/appendix/before}{\def\cref@section@alias{appendix}}
\AddToHook{cmd/appendix/before}{\def\cref@subsection@alias{appendix}}
\makeatother

\newcommand{\yes}{\textcolor{ForestGreen}{\checkmark}}
\newcommand{\no}{\textcolor{BrickRed}{\cross}}

\newcommand{\EK}[2]{E_K^{#1 \to #2}}
\newcommand{\SK}[1]{\smash{\sigma_K^{(#1)}}}

\newcommand{\EL}[2]{E_L^{#1 \to #2}}
\newcommand{\SL}[1]{\smash{\sigma_L^{(#1)}}}

\newcommand{\EKL}[2]{E_{KL}^{#1 \to #2}}
\makeatletter
\def\thmt@storefile{thmstore}
\makeatother

\usepackage{ifthen}
\newboolean{showcalc}
\setboolean{showcalc}{false}  

\newcommand{\probPlot}{
    \hspace{-2cm}
    \begin{axis}[
    axis lines=none,
    xlabel=$x$, ylabel=$y$,
    xmin=1.25, xmax=8,
    ymin=-1, ymax=8,
    domain=0.34:7,
    samples=200
    ]
    \addplot[Blue!20, fill] 
    ({x}, -{1/x}+2.9) 
    \closedcycle;
    \end{axis}
}

\newcommand{\probFig}{
  \begin{tikzpicture}[scale=1.2]

    \node[] at (0,0) {\probPlot};
    \node[yscale = -1] at (0,4.05) {\probPlot};
    \node[Blue] at (4.3,1) {$\Sigma^{\mathrm{rand}}$};
    
    \foreach [count=\i] \x in {0.2,0.8,1.4,2,2.6,3.2,3.8} {
        \pgfmathsetmacro{\yval}{{-1/(2.8*(\x+1))+2.025}}  
        \draw[BrickRed] (\x, 2.025) -- (\x,\yval); 
        \node[draw, fill=black, black, circle, minimum size=0.1cm, inner sep=0pt] (sigmaM) at (\x,2.025) {};
        \node[draw, fill=Blue, Blue, circle, minimum size=0.1cm, inner sep=0pt] 
            at (\x,\yval) {};
    }
    \node[] at (4.4,2.025) {$\dots$};

    \foreach [count=\i] \x in {0.2,0.8,1.4} {
        \pgfmathsetmacro{\yval}{{-1/(2.7*(\x+0.8))+2.025}}  
        \node[font = \tiny] at (\x+0.1,2.225) {$\rho^{(\i)}$};
    }

    \pgfmathsetmacro{\yval}{{-1/(2.8*(1.7+0.3))+2.025}} 
    \node[font = \tiny] at (2,2.225) {$\dots$};

  \end{tikzpicture}

}

\newcommand{\QFig}{
  \begin{tikzpicture}[scale=1.2]

    \draw[fill = Blue!20!white, Blue!20!white] (0,-0) rectangle (5,1.1);
    \draw[fill = Blue!20!white, Blue!20!white] (0,1.9) rectangle (5,3);

    \foreach \x in{0.2,0.8,1.4,2,2.6,3.2,3.8,4.4} {
        \pgfmathsetmacro{\yval}{1.09}   
        \draw[BrickRed] (\x, 1.5) -- (\x,\yval);
        \node[draw, fill=black, black, circle, minimum size=0.1cm, inner sep=0pt] at (\x,1.5) {};
        \node[draw, fill=Blue, Blue, circle, minimum size=0.1cm, inner sep=0pt] 
            at (\x,\yval) {};
    }

    \pgfmathsetmacro{\yval}{1.09}
    \foreach [count=\i] \x in {0.2,0.8,1.4} {
        \node[font = \tiny] at (\x+0.1,1.7) {$\rho^{(\i)}$};
    }

    \node[] at (4.9,1.5) {$\dots$};
    \node[Blue] at (4.5,0.5) {$\Sigma^{\mathrm{rand}}$};
    \node[font = \tiny] at (2.03,1.7) {$\dots$};

  \end{tikzpicture}
}
\newcommand{\lockandkey}[4]{
    \begin{tikzpicture}[baseline=-0.5ex]
        \draw[thick] (-0.7,-0.4) rectangle (0.7,0.4);
        \node[scale = 0.7] at (-0.5,-0.2) {#4};
        \node[scale = 1] at (0,0) {#1};
        \draw[-stealth, double] (0,0.8) node[above, align = center] {#2} -- (0,0.4);
        \draw[-stealth, double] (0,-0.4) -- (0,-0.8) node[below, align = center] {#3};
    \end{tikzpicture}
}

\newcommand{\lockandkeynoinput}[4]{
    \begin{tikzpicture}[baseline=-0.5ex]
        \draw[-stealth, double, transparent] (0,0.8) node[above, align = center] {#2} -- (0,0.4);
        \draw[thick] (-0.7,-0.4) rectangle (0.7,0.4);
        \node[scale = 0.7] at (-0.5,-0.2) {#4};
        \node[scale = 1] at (0,0) {#1};
        \draw[-stealth, double] (0,-0.4) -- (0,-0.8) node[below, align = center] {#3};
    \end{tikzpicture}
}

\newcommand{\lockandkeynooutput}[4]{
    \begin{tikzpicture}[baseline=-0.5ex]
        \draw[-stealth, double] (0,0.8) node[above, align = center] {#2} -- (0,0.4);
        \draw[thick] (-0.7,-0.4) rectangle (0.7,0.4);
        \node[scale = 0.7] at (-0.5,-0.2) {#4};
        \node[scale = 1] at (0,0) {#1};
        \draw[-stealth, double, transparent] (0,-0.4) -- (0,-0.8) node[below, align = center] {#3};
    \end{tikzpicture}
}

\newcommand{\key}[3]{
    \lockandkey{$#1$}{$#2$}{$#3$}{$K$}
}

\newcommand{\lock}[3]{
    \lockandkey{$#1$}{$#2$}{$#3$}{$L$}
}

\newcommand{\wireKtoL}[5]{
    \begin{tikzpicture}[baseline=-0.5ex]
       \draw[double,-stealth]
            plot[smooth] coordinates {(-1.2,-0.4)(-1.2,-0.5)(-1.1,-0.6) (-0.8,-0.8)(-0.5,-0.7) (0,0.6) (0.8,1) (1,0.6)}
            -- (1,0.46); 
        \node[] at (-1.2,0) {\lockandkeynooutput{$#1$}{$#2$}{$#3$}{$K$}};
        \node[] at (0,0) {$\otimes$};
        \node[] at (1,0) {\lockandkeynoinput{$#4$}{\phantom{$#2$}}{$#5$}{$L$}};
    \end{tikzpicture}   
}

\newcommand{\NSbox}[6]{
    \begin{tikzpicture}[baseline=-0.5ex]
        \draw[-stealth, double] (-1.2,0.8) node[above, align = center] {$#1$} -- (-1.2,0.4);
        \draw[-stealth, double] (0,0.8) node[above, align = center] {$#2$} -- (0,0.4);
        \draw[-stealth, double] (1.2,0.8) node[above, align = center] {$#3$} -- (1.2,0.4);
        \draw[thick] (-1.6,-0.4) rectangle (1.6,0.4);
        \draw[dotted] (-0.6,-0.4) -- (-0.6,0.4);
        \draw[dotted] (0.6,-0.4) -- (0.6,0.4);
        \draw[-stealth, double] (-1.2,-0.4) -- (-1.2,-0.8) node[below, align = center] {$#4$};
        \draw[-stealth, double] (0,-0.4) -- (0,-0.8) node[below, align = center] {$#5$};
        \draw[-stealth, double] (1.2,-0.4) -- (1.2,-0.8) node[below, align = center] {$#6$};
    \end{tikzpicture}
}

\title{\LARGE{Against probability:} \\
\Large{A quantum state is more than a list of probability distributions}}
\date{}
\author{Ladina Hausmann and Renato Renner}
\affil{\small \it Institute for Theoretical Physics, ETH Zürich, 8093 Z\"urich, Switzerland}

\begin{document}
\maketitle
\vspace{-1.6cm}

\begin{abstract}
  The state of a quantum system can be represented by listing the outcome probabilities for a tomographically complete set of measurements. Such representations appear throughout physics, for example, in quantum field theory via correlation functions and in quantum foundations within generalized probabilistic frameworks. In this paper, we show a no-go result: To enable useful statements, the probability representation must be topologically robust---preserving the notion of closeness between states. Yet, a topologically robust probability representation cannot simultaneously retain other essential structure, such as the subsystem structure.
\end{abstract}

\section{Introduction}
Let $P_M(\rho)$ denote the probability distribution obtained by applying a measurement $M$ to a quantum state  $\rho$ on a separable Hilbert space $\mathcal{H}$. If one collects these probability distributions for all $M$ from a tomographically complete set $\mathcal{M}$ of measurements, the resulting tuple 
\begin{equation}
    \mathbf{P}_{\mathcal{M}}(\rho) \coloneq \bigl(P_M(\rho)\bigr)_{M \in \mathcal{M}}
\end{equation}
uniquely specifies the quantum state $\rho$. Such \emph{probability representations} arise, for instance, when considering correlation functions in quantum field theories \cite{Peskin1995}, and are widespread in quantum information theory and quantum foundations~\cite{Wootters1986,Hardy_2001,Mana2004,Barrett2005a,Barrett_2007,Dakic2009,Wilce2009,Chiribella2010,Chiribella2011,Masanes_2011,Appleby2011,Zapo2012,delaTorre2012,Janotta2013,Masanes2013,Masanes2014b,Barnum2014,Hardy2015,Muller2016,Muller2016,Hohn2017,Hohn2017tool,Hohn_2017rules,DAriano_2017,Selby2021,Plavala_2023,Fuchs_2013}. 

This raises a fundamental question: does $\mathbf{P}_{\mathcal{M}}(\rho)$ represent $\rho$ faithfully? If faithfulness merely meant injectivity, the answer would be yes, as this is precisely the notion of tomographic completeness. However, a faithful representation should also ensure that statements derived from $P_{\mathcal{M}}(\rho)$ are robust under small perturbations. In this work, we explain why robustness is necessary, and why it nevertheless cannot be achieved for probability representations without losing essential structure, such as the subsystem structure.

\section{Why robustness is necessary}
We begin with an example where the use of probability representations appears natural: generating random bits $R_1, \ldots, R_\ell$ from a quantum process. For concreteness, imagine a protocol that produces each bit $R_i$ as follows: prepare $n$ unstable atoms and count the number of decays within a fixed time interval (e.g., one second). Set $R_i=0$ if this number is even and $R_i=1$ otherwise. 

An $\ell$-bit string $R = R_1 \cdots R_\ell$ is said to be \emph{perfectly random} if it cannot be predicted with probability greater than~$2^{-\ell}$, even by an all-powerful agent with access to any side information $E$ available prior to the process, such as the internal state of the source supplying the atoms. Formally, this means that the joint state $\smash{\rho^{(n)} = \rho^{(n)}_{RE}}$ of~$R$ and~$E$, when each $R_i$ was generated using $n$ atoms, satisfies
\begin{equation} \label{eq:perfectrandomness}
    \rho^{(n)} \in \Sigma^{\mathrm{rand}} \coloneq \left\{\bar\Pi^{(\ell)}_R \otimes \sigma_E \Big| \ell \in \mathbb{N}, \sigma_E  \in \mathcal{D}(E)\right\}
\end{equation}
where $\bar \Pi^{(\ell)}_R$ is the uniform state on bit strings of length $\ell$, and $\mathcal{D}(E)$ is the set of density operators on $E$. 

In any practical randomness-generation protocol, the output $R$ inevitably has residual correlations with other systems. Moreover, since $R$ is produced by a quantum process, it is not automatically classical: in general, the joint state $\smash{\rho^{(n)}_{R E}}$ is not a cq-state. For these reasons, condition \eqref{eq:perfectrandomness} cannot be achieved exactly, but only approached asymptotically by investing additional resources~\cite{Vazirani2012,Pironio2013,Frauchiger2013}, for instance by increasing the number $n$ of atoms in our example protocol. Therefore, it is standard to adopt an approximate condition of the form 
\begin{equation}\label{eq:state_conv}
    \lim_{n \to \infty} \delta(\rho^{(n)}, \Sigma^{\mathrm{rand}}) = 0, 
\end{equation}
where $\delta(\rho^{(n)}, \Sigma^{\mathrm{rand}}) \coloneq \inf_{\sigma \in \Sigma^{\mathrm{rand}}} \delta(\rho^{(n)}, \sigma)$, with $\delta(\cdot, \cdot)$ denoting the trace distance. Operationally, this quantifies the maximal probability with which $\rho^{(n)}$ can be distinguished from the ideal behaviour defined by $\Sigma^{\mathrm{rand}}$ when one has access to both $R$ and $E$~\cite{Nielsen2010,Ferradini_2025}. 
We illustrate  the definition of randomness in \eqref{eq:state_conv} with an example that does \emph{not} satisfy it.
\begin{example*}\label{ex:anti_symmetric}
    Let $n = \ell$, $R$ and $E$ be $2^n$-dimensional systems, and the joint state of~$R$ and $E$ after the randomness generation protocol 
    \begin{align}
    \rho^{(n)}_{RE} \propto \sum_{1\leq u < v \leq 2^n} \pi_{u,v}
    \end{align}
    where $\pi_{u,v}$ denotes the projector on the subspace of $R E$ spanned by ${\ket{u}_R \ket{v}_E - \ket{v}_R \ket{u}_E}$, for an arbitrary choice of orthonormal bases.  This state is entangled and thus distinct from the states in $\Sigma^{\mathrm{rand}}$, which are separable. Concretely, as shown in \cite[Example II.9.]{Christandl_2007}, 
        \begin{equation}\label{eq:not_conv}
            \forall n \in \mathbb{N}: \delta(\rho^{(n)}, \Sigma^{\mathrm{rand}}) \geq \frac{1}{4}.
        \end{equation}
    Therefore, $R$ is not approximately random according to criterion \eqref{eq:state_conv}.
\end{example*}

To work at the level of probability representations, two conditions must be met: (i)~Approximate relations like~\eqref{eq:state_conv} must be expressible via a distance measure defined directly on probability distributions, such as the statistical distance. (ii)~The representation $\mathbf{P}_\mathcal{M}$ must have enough structure to support definitions like~\eqref{eq:perfectrandomness}.  In particular, the subsystem structure of $RE$ must be preserved, i.e., $\mathcal{M}$ only contains measurements that act locally on~$R$ and $E$. Satisfying both conditions is non-trivial~\cite{Koenig2007}, as we now illustrate by continuing the above example.

\begin{figure}[tbp]
    \centering
    \begin{subfigure}{0.35\textwidth}
         \begin{tikzpicture}[scale=1]
            \node[] at (0,0) {\QFig};
        \end{tikzpicture}
        \caption{}
        \label{fig:state}
    \end{subfigure}
    \hspace{2cm}
    \begin{subfigure}{0.35\textwidth}
        \hspace{-0.5cm}
        \begin{tikzpicture}[scale=1]
            \node[] at (0,0) {\probFig};
        \end{tikzpicture}
        \vspace{-1.25cm}
        \caption{}
        \label{fig:prob}
    \end{subfigure}

    \caption{{\bf Non-robustness of $\mathbf{P}_{\mathcal{M}_{\otimes}}$.} The figure illustrates how the example violates \cref{def:preservation}.  A region of the quantum state space is shown, containing the sequence of states $\rho^{(n)}$ (black) and the corresponding closest points in $\Sigma^{\mathrm{rand}}$ (blue).  Panel~(a) uses the trace distance~$\delta$, which stays constant for all $n$. Panel~(b) uses the metric~$d_{\mathcal{M}_{\otimes}}$ induced by the local representation $\mathbf{P}_{\mathcal{M}_{\otimes}}$, for which the distance to $\Sigma^{\mathrm{rand}}$ shrinks with increasing~$n$.
    }\label{fig:prod_conv}
\end{figure} 

\begin{example*}[continued]
    While the following holds for any local measurement,\footnote{To see this, it suffices to observe that the state on $E$ conditioned on any outcome of a rank-$1$ measurement applied to $R$ is maximally mixed on a subspace of dimension $\smash{2^n-1}$, and thus is $\smash{2^{-n}}$-close to~$\bar\Pi^{(n)}_E$.} we focus, for illustrative purposes, on the case where the same rank-$1$ measurement~$\bar{M}$ is applied to both $R$ and $E$. Exploiting the antisymmetry of $\smash{\rho^{(n)}_{RE}}$, one finds that the joint probability distribution $P_{\bar{M} \otimes \bar{M}} = \smash{P_{\bar{M} \otimes \bar{M}}(\rho^{(n)}_{RE})}$  of the two outcomes is
    \begin{equation}
        P_{\bar{M} \otimes \bar{M}}(x, y) = \begin{cases}
            \quad 0 & \text{if } x = y \\
            \frac{1}{2^n(2^n-1)} & \text{else}.
        \end{cases}
    \end{equation}
    This implies that $\smash{\frac{1}{2}\|P_{\bar{M} \otimes \bar{M}}(\rho^{(n)}_{RE}) - P^{(n)} \times P^{(n)} \|_1 \leq 2^{-n}}$,
    where $P^{(n)}$ is the uniform distribution on an alphabet of size $2^n$. Consequently, at the level of such probability representations, $\smash{\rho^{(n)}_{RE}}$ approaches, as $n \to \infty$, the set $\Sigma^{\mathrm{rand}}$ defining perfect randomness, see \cref{fig:prod_conv}.
\end{example*}

The conclusion of this example can be phrased concisely by introducing the metric
\begin{align} \label{eq:representationmetric}
    d_{\mathcal{M}}(\rho, \sigma) \coloneq \frac{1}{2} \sup_{M \in \mathcal{M}} \bigl\| P_M(\rho) - P_M(\sigma)\bigr\|_1 \, .
\end{align}
Namely, choosing $\mathcal{M} = \mathcal{M}_{\otimes}$ to be the set of local measurements, the sequence $\smash{(\rho_{RE}^{(n)})_{n \in \mathbb{N}}}$ converges to $\Sigma^{\mathrm{rand}}$ with respect to $d_{\mathcal{M}_{\otimes}}$. Yet the same is not true for the metric $\delta$; see \eqref{eq:not_conv}. This motivates the following definition. 

\begin{definition}\label{def:preservation}
    We say that the representation $\mathbf{P}_{\mathcal{M}}$ is \emph{topologically robust} (or simply \emph{robust}) if the induced metric~$d_{\mathcal{M}}$ satisfies 
    \begin{equation}\label{eq:implication}
       \lim_{n \to \infty} d_{\mathcal{M}}(\rho^{(n)},\Sigma) = 0 \ \implies \ \lim_{n \to \infty} \delta(\rho^{(n)}, \Sigma) = 0
    \end{equation}
    for all subsets~$\Sigma$ and sequences $(\rho^{(n)})_{n \in \mathbb{N}}$ of states.
\end{definition}

Note that the opposite implication is always true. Thus, robustness implies that it does not matter whether the metric $d_{\mathcal{M}}$ or $\delta$ is chosen for a definition like \cref{eq:state_conv}. 
However, the above example shows that robustness does not always hold, in particular, it fails for~$\mathcal{M} = \mathcal{M}_{\otimes}$.

To quantify the success of an information processing protocol, the relevant metric should characterize how well the protocol output can be distinguished from the ideal behaviour. Both metrics $\delta$ and $d_{\mathcal{M}}$ have this property. The difference is that the first metric quantifies distinguishability under arbitrary measurements and the second under the restricted set~$\mathcal{M}$. 

The non-robustness of $d_{\mathcal{M}_{\otimes}}$ raises the question: which set of measurements is more operationally relevant? To decide this, we use the \emph{principle of composability}~\cite{MaurerRenner2011}. The principle demands that the underlying distance measure~$d$ be stable under the addition of auxiliary systems in any state $\Psi$, i.e.,
\begin{equation} \label{eq:stability} 
    d(\rho_{RE}, \Sigma_{RE}) = d(\rho_{RE} \otimes \Psi_{R' E'}, \Sigma_{RE} \otimes \Psi_{R' E'}).
\end{equation}  
While the trace distance $\delta$ satisfies this principle, the metric $d_{\mathcal{M}_{\otimes}}$ defined with respect to product measurements acting separately on $RR'$ and $EE'$, does not. The former is a consequence of the monotonicity of the trace distance under data-processing \cite[Theorem~9.2]{Nielsen2010}. For the latter, note that our example implies $d_{\mathcal{M}_{\otimes}}(\rho,\Sigma) \ll \delta(\smash{\rho}, \smash{\Sigma})$; yet, it was proven in~\cite{Vaidman2003} that ${d_{\mathcal{M}_{\otimes}}(\rho \otimes \Psi, \Sigma \otimes \Psi)} \approx {\delta(\rho \otimes \Psi, \Sigma \otimes \Psi)}$ when $\Psi $ contains sufficient entanglement.

To see the operational significance of the composability principle~\eqref{eq:stability}, consider an agent, Alice, with access to systems~$RR'$, and another agent, Eve, with access to side information~$EE'$. Suppose that $R$ is approximately random when ignoring~$R'E'$, i.e., there is a state~$\sigma_{RE} \in \Sigma^{\mathrm{rand}}$ indistinguishable from~$\rho_{RE}$ by any measurement. Then criterion~\eqref{eq:stability} applied to $\delta$ ensures that this indistinguishability is preserved in the full description including~$R'E'$. However, no such guarantee would hold if randomness solely required indistinguishability under product measurements, as $d_{\mathcal{M}_{\otimes}}$ violates the composability principle~\eqref{eq:stability}.

We have thus answered the question posed at the outset: injectivity alone does not suffice to ensure that a representation $\mathbf{P}_{\mathcal{M}}$ is faithful; robustness in the sense of \cref{def:preservation} is also required. Without this property, approximate statements established at the level of the representation cannot, in general, be ``pulled back'' to the level of density operators.\footnote{This problem does not affect exact statements. These can be proved at the level of the probability representation and then directly pulled back to density operators. A beautiful example is the proof of the quantum de Finetti theorem for infinitely exchangeable states proposed in~\cite{Caves_2002}.} This is illustrated by the commuting diagram in \cref{fig:commuting_diag}.

\begin{figure}[tbp]
  \centering
  \begin{tikzpicture}[
    every node/.style={font=\large, text centered},
    scale=0.65
]

  \node (rho) at (-3,1.5) {$\rho$};
  \node[] (app) at (7.5,1.5){$\delta(\rho,\Sigma) \approx 0$};
  \node[] (pm) at (-3,-1.5) {$\mathbf{P}_{\mathcal{M}}(\rho)$};
  \node[] (app2) at (7.5,-1.5){$d_{\mathcal{M}}(\rho,\Sigma) \approx 0$};
  \node[BrickRed, thick] (q) at (7.5,0) {\Large{\textbf{?}}};

  \draw[-Implies, double, draw=Plum, thick,line width=0.4mm, double distance=1.5pt] (-2,1.5) -- node[above, text=Plum, font=\normalsize]{QIT arguments} (5.6,1.5);
  \draw[-Implies, double, draw=RedOrange, thick,line width=0.4mm, double distance=1.5pt] (-2,-1.5) -- node[above, text=RedOrange, font=\normalsize]{Probability arguments} (5.6,-1.5);

  \draw[line width=0.4mm, stealth-stealth] (pm) -- (rho);
  \draw[line width=0.4mm, Implies-, double,  double distance=1.5pt] (app) -- (q);
  \draw[line width=0.4mm, double,  double distance=1.5pt] (app2) -- (q);

\end{tikzpicture}
\caption{{\bf State-space vs.~representation-space approximations.} The diagram illustrates the requirement that approximate statements established by applying probability arguments to the representation remain valid when ``pulled back'' to density operators. Topological robustness (see \cref{def:preservation}) guarantees this.}\label{fig:commuting_diag}
\end{figure}

\section{A topological characterization of robustness}\label{sec:topo_problem}
\Cref{def:preservation} concerns the convergence of sequences and, therefore, has a topological character. One might expect that a failure of robustness thus means that the topologies induced by the metrics $d_{\mathcal{M}}$ and $\delta$ are inequivalent; in other words, that $\mathbf{P}_{\mathcal{M}}$ is not a homeomorphism. 
This intuition is only partially correct: In fact, the topologies induced by $d_{\mathcal{M}}$ and~$\delta$ are identical on $\mathcal{D}(\mathcal{H})$, even when~$\mathbf{P}_{\mathcal{M}}$ is not robust, except in pathological cases. The pathological cases are those where $\mathcal{M}$ is fragile, i.e., the topology induced by $d_{\mathcal{M}}$ changes under the addition of a single additional measurement.\footnote{See \cref{def:fragile} in the Appendix for the precise definition.} The set of local measurements, in particular, is not fragile. We refer to the Appendix for the proofs of all statements.
\begin{restatable}[]{proposition}{proptopo}\label{prop:topo}
    The topologies induced by $d_{\mathcal{M}}$ and~$\delta$ are identical on the space of density matrices $\mathcal{D}(\mathcal{H})$, except when $\mathcal{M}$ is fragile.
\end{restatable}

Despite this result, robustness can be characterized topologically. Crucially, one has to consider the vector space $\smash{\mathrm{span}(\mathcal{D}(\mathcal{H}))}$ spanned by the density operators, which we equip with the norm $\|\cdot\|_{\mathcal{M}} \coloneq \sup_{M \in \mathcal{M}}\|P_M(\cdot)\|_{1}$; see~\cref{fig:topoproblem}. 

\begin{restatable}[]{proposition}{topology}\label{prop:topology}
    The following statements are equivalent:
    \begin{enumerate}[label = (\arabic*), itemsep=0pt, topsep=3pt]
        \item\label{item:robust} The representation $\mathbf{P}_{\mathcal{M}}$ is robust.
        \item\label{item:continuous} The topologies induced by $\|\cdot\|_{\mathcal{M}}$ and~$\|\cdot\|_{1}$ are identical on $\mathrm{span}(\mathcal{D}(\mathcal{H}))$.\footnote{This is equivalent to requiring that the inverse of the linear extension of the map $\mathbf{P}_{\mathcal{M}}$ (corresponding to the arrow labelled with ``?'' in \cref{fig:commuting_diag}) is continuous with respect to the norm on the probability representation induced by $\| \cdot \|_{\mathcal{M}}$ and the  trace norm $\| \cdot \|_1$ on  $\mathrm{span}(\mathcal{D}(\mathcal{H}))$.}
        \item\label{item:complete} $\mathrm{span}(\mathcal{D}(\mathcal{H}))$ is complete with respect to $\|\cdot\|_{\mathcal{M}}$.
    \end{enumerate}
\end{restatable}

\begin{figure}
    \centering
    \begin{tikzpicture}[scale=1.5]

        \draw[thick, ->] (0,0) -- (0,2) node[left] {};
        \draw[thick, ->] (0,0) -- (4.5,0) node[below] {};
        
        \draw[line width=8pt, OliveGreen, line cap=round,opacity=0.4] (-0.1,1.8) -- (3.8,-0.1);
        \draw[latex-, thick, OliveGreen] (1.65,1) -- (2,1.5);
        \node[OliveGreen,align=center] at (2.5,2) {Open set containing \\ the state space};

        \draw[very thick, Blue] (-0.1,1.8) -- (3.8,-0.1);
        \draw[latex-, thick, Blue] (1.55,1) -- (1,0.6);
        \node[Blue,align=center] at (1,0.5) {State space};

    \end{tikzpicture}
    \caption{{\bf Topological characterization of robustness.} Even if a representation $\mathbf{P}_{\mathcal{M}}$ fails to be robust, the topologies induced by the metrics $d_{\mathcal{M}}$ and $\delta$ are generally identical on the state space $\mathcal{D}(\mathcal{H})$ (blue line). However, if we extend the consideration to an open set around $\mathcal{D}(\mathcal{H})$ (green) within $\mathrm{span}(\mathcal{D}(\mathcal{H}))$, then the equivalence between the topologies induced by the norms $\| \cdot\|_{\mathcal{M}}$ and $\|\cdot\|_1$ serves as a criterion for the robustness of $\mathbf{P}_{\mathcal{M}}$ (see \cref{prop:topology}).} \label{fig:topoproblem}
\end{figure}

\section{Robustness is incompatible with structure}
The example discussed above implies that representations based on the set of local measurements $\mathcal{M}_{\otimes}$ are not robust. This raises the question of whether the issue can be avoided by relaxing the locality condition. As we show below, the answer is negative: non-robustness is in fact a generic feature of any representation that has structure. 

To make this precise, we quantify the ``structure'' of a representation $\mathbf{P}_\mathcal{M}$ using tools from information theory, in particular an appropriate notion of entropy. Entropy characterizes the minimum size to which data can be compressed. This requires an encoder, which compresses the data, and a decoder, which retrieves it. We require these to be selected from a set of ``physically allowed'' operations. A standard result in quantum information is that these must be completely positive maps. Therefore, the concatenation of $\mathbf{P}_{\mathcal{M}}$ and the encoder takes every $\rho$ to a compressed list $(P_E(\rho))_{E \in \mathcal{O}}$ indexed by a set $\mathcal{O}$. Each entry is the output of a linear map and, thus, of the form $P_E(\rho) = \tr(E \rho)$, for an appropriately chosen positive operator $E$, which we use as the index. Similarly, the decoder for a measurement~$M$ is a convex map of the form
\begin{equation}
    \mathcal{D}_M: (P_{E})_{E \in \mathcal{O}} \mapsto \big(\sum_{E} p^{(M_i)}_E P_{E} \big)_{M_i \in M},
\end{equation}
where $p_E^{(M_i)} \geq 0$ for any effect $M_i$ of $M$. For any $M \in \mathcal{M}$, one wants the output of $\mathcal{D}_M$ to reproduce (up to a small error $\varepsilon$) the probability distribution~$P_M(\rho)$ of the outcomes when measuring~$\rho$. This leads to the following definition.\footnote{See the Appendix for the formal definition, which also applies to continuous measurements.}

\begin{definition}[Informal]\label{def:epsilon_decoded}
    Let $M$ be a measurement and~$\mathcal{O}$ a set of positive operators. We say that $P_M$ can be $\varepsilon$-decoded from~$\mathcal{O}$ if there exists a decoding operation~$\mathcal{D}_M$ such that 
    \begin{equation}
        \forall \rho \in \mathcal{D}(\mathcal{H}): \ \bigl\|\mathcal{D}_M\bigl(\bigl(\tr(E \rho)\bigr)_{E \in \mathcal{O}}\bigr) - P_M(\rho)\bigr\|_1 < \varepsilon.
    \end{equation}
\end{definition}

The length of the compressed list grows linearly in $|\mathcal{O}|$. Therefore, we define the \emph{entropy of a representation $\mathbf{P}_{\mathcal{M}}$} as $H^{\varepsilon}(\mathbf{P}_{\mathcal{M}}) \coloneq |\mathcal{O}|$, where we take the minimum set $\mathcal{O}$ required to $\varepsilon$-decode all $P_M$ with $M \in \mathcal{M}$. For infinite-dimensional Hilbert spaces this quantity is infinite. In this case, we consider its scaling for restricted measurement sets $\mathcal{M}|_{\Pi}$, obtained by preceding each measurement in $\mathcal{M}$ by an orthogonal finite-rank projector~$\Pi$.
\begin{definition}\label{def:asymptotic_entropy}
    Let $\{\Pi_{2^n}\}_{n \in \mathbb{N}}$ be a nested family of projectors with $\mathrm{rank}(\Pi_{2^n}) \geq 2^n$, and~$(\varepsilon_n)_{n \in \mathbb{N}}$ a zero-sequence. We say the \emph{asym\-ptotic entropy of $\mathbf{P}_\mathcal{M}$ is at most as large as the sequence $(H_n)_{n \in \mathbb{N}}$}, if $H^{\varepsilon_n}(\mathbf{P}_{\mathcal{M}|_{\Pi_{2^n}}}) \leq H_n$. 
\end{definition}

A calculation shows that the asymptotic entropy of any representation $\mathbf{P}_\mathcal{M}$ is upper bounded by approximately $\smash{(2^{2^{n}})_{n \in \mathbb{N}}}$. Non-negligible structure in $\mathbf{P}_\mathcal{M}$ manifests itself in an entropy below this upper bound. This characterization of structure allows us to state our no-go result. 

\begin{restatable}{theorem}{minentropy}\label{thm:entropy}
    If $\mathbf{P}_\mathcal{M}$ has asymptotic entropy at most $(2^{\delta_n 2^n})_{n \in \mathbb{N}}$ for a zero-sequence~$(\delta_n)_{n \in \mathbb{N}}$, then the representation $\mathbf{P}_{\mathcal{M}}$ is not robust. 
\end{restatable}

The theorem tells us that a probability representation $\mathbf{P}_{\mathcal{M}}$ cannot be robust when it has non-negligible structure. This is the case for~$\mathbf{P}_\mathcal{M_{\otimes}}$, where a straightforward calculation shows that the asymptotic entropy approximately scales only as $\smash{(2^{2^{n/2}})_{n\in\mathbb{N}}}$. This confirms the intuition that the representation based on local measurements has significant structure and, thus, reproduces the conclusion of the example.\footnote{Alternatively, the statement can also be derived from well-known data hiding results \cite{Terhal2001,DiVincenzo2004,Eggeling2002,Hayden_2004,Koenig2007,Matthews2009,Lancien2013,Chitambar2014,Aubrun2015,Lami2021,Correa2022,Cheng_2023,Mele2025,Ha2025}.}

\begin{restatable}[]{corollary}{prodConv}\label{cor:prod_conv}
    The representation $\mathbf{P}_{\mathcal{M_{\otimes}}}$ is not robust.
\end{restatable} 
The same approximate scaling of  $\smash{(2^{2^{n/2}})_{n\in\mathbb{N}}}$ holds for the entropy of representations based on a larger set of measurements than $\mathcal{M}_{\otimes}$, like the set of measurements that can be implemented via local operations and classical communication. Furthermore, because this scaling lies far below the threshold in \cref{thm:entropy}, the corollary is stable: non-robustness persists even if one adds a large number of extra measurements to $\mathcal{M_{\otimes}}$. 

We now present another example of structure in a representation $\mathbf{P}_{\mathcal{M}}$, arising from the principle of minimality applied to $\mathcal{M}$. Notable instances include minimal tomographically complete measurement sets such as SIC-POVMs~\cite{Renes_2004} or fiducial measurements in the framework of generalized probabilistic theories~\cite{Wootters1986,Hardy_2001,Barrett_2007,Masanes_2011,Janotta2013,Hardy2015,Muller2016,Fuchs_2010,Appleby2011}. For a minimal description of a system, one also chooses the smallest Hilbert space that still contains all relevant degrees of freedom. For example, if one uses an atom as a qubit in a quantum computer, then one works with the corresponding $2$-dimensional subspace rather than the high-dimensional Hilbert space of the atom. 

While this minimization of the dimension is often implicit, it involves a non-trivial consistency condition: restricting from a $D$-dimensional to a $d$-dimensional Hilbert space should not weaken the representation. This consistency condition may be expressed formally as follows: using \cref{def:epsilon_decoded} with $\varepsilon = 0$, we require that $\mathbf{P}_{\mathcal{M}_D|_{\Pi_d}}$ is decodable from the effects of the measurements in $\mathcal{M}_d$, where~$\Pi_d$ is the projector onto the $d$-dimensional subspace. 

We say that a representation $\mathbf{P}_{\mathcal{M}}$ is \emph{efficient} if $\mathcal{M}$ admits a decomposition $\mathcal{M} = \bigcup_{d} \mathcal{M}_d$ into tomographically complete measurement sets $\mathcal{M}_d$ satisfying the above consistency condition, such that the number of distinct effects contained in $\mathcal{M}_d$ grows at most polynomially in~$d$. This choice is motivated by the fact that, for each $d$, the dimension of the state space is at most quadratic, and hence polynomial, in~$d$. 
The asymptotic entropy of any efficient representation, thus, grows at most polynomially in $2^n$. This yields another corollary to \cref{thm:entropy}, which is also stable in the same sense as \cref{cor:prod_conv}.

\begin{corollary}\label{cor:non-redundant}
   If the representation $\mathbf{P}_{\mathcal{M}}$ is efficient, then $\mathbf{P}_{\mathcal{M}}$ is not robust. 
\end{corollary}

Taken together, these results show that the use of probability representations leads to a fundamental dilemma, which is summarized in \cref{tab:result}. 

\begin{table}
    \centering
    \renewcommand{\arraystretch}{0.8}

    \begin{tabular}{@{}lcccccc@{}}
    \toprule
    \hspace*{3.8cm} & & & & \multicolumn{3}{c}{Structure} \\
    & & {Robust}
    & & {Subsystems} 
    & & {Efficient} \\ 
    \midrule

    \multicolumn{3}{l}{{Density operator}} & & \hspace*{2.9cm} & & \hspace*{2.9cm} \\ 
    \hspace{3em} $\mathcal{D}(\mathcal{H})$ && $\yes$ && $\yes$ && $\yes$ \\ 

    \multicolumn{3}{l}{{Probability representations}}  \vspace{1pt}  && &&  \\ 
    \hspace{3em} $\mathbf{P}_{\mathcal{M}_{\mathrm{all}}}$ && $\yes$ && $\no$ && $\no$ \\ 

    \hspace{3em} $\mathbf{P}_{\mathcal{M}_{\otimes}}$ && $\no$ && $\yes$ && $\no$ \\ 

    \hspace{3em} $\mathbf{P}_{\mathcal{M}_{\mathrm{efficient}}}$  && $\no$ && \makecell[c]{depends on $\mathcal{M}_{\mathrm{efficient}}$} && $\yes$ \\ 
    \bottomrule
    \end{tabular}
    \caption{
    {\bf Density operator vs.~probability representations.} The density operator representation of a quantum state is robust in the sense that small deviations in the representation are physically insignificant. Furthermore, it respects the subsystem structure, and it is efficient, requiring only a few real numbers. In contrast, probability representations that are topologically robust cannot have any non-negligible structure. \label{tab:result} }
\end{table}

\section{Beyond quantum}
In quantum foundations, one often studies generalizations of quantum theory where states need not be representable by density operators. Probability representations are well suited for this task, as they allow one to directly modify the constraints that quantum theory imposes on the admissible lists of outcome probabilities. A widely used framework based on this idea is that of generalized probabilistic theories (GPTs)~\cite{Hardy_2001,Barrett2005,Barrett2005a,Barrett_2007,Spekkens2005,Chiribella2010,Chiribella2011,Hardy2011,Janotta_2014,DAriano_2017,Plavala_2023}, where each GPT is specified by the set of probability assignments corresponding to its valid states. 

Since probability representations are fundamental to the GPT framework, the issues summarized in \cref{tab:result} also pose a challenge for GPTs. In fact, the situation is even more severe, as \cref{prop:topo} does not generalize to GPTs beyond quantum theory. 
To state this result, we extend the trace distance to arbitrary GPTs by defining $\delta(\rho, \sigma) \coloneq \frac{1}{2} \sup_{M} \|P_M(\rho) - P_{M}(\sigma)\|_1$ where the supremum ranges over all measurements $M$. This metric generalizes the quantum trace distance, in the sense that it quantifies the probability with which two states can be distinguished. Note that, as in the quantum case, we allow for state spaces of unbounded dimension.

\begin{theorem}\label{thm:GPT} 
    There exists a GPT for which the topologies induced by $d_{\mathcal{M}_{\otimes}}$ and $\delta$ are different on state space, despite $\mathcal{M}_{\otimes}$ being not fragile.
\end{theorem}

\section{Conclusion}

We conclude with a discussion of the implications of our results for various areas of research that use probability representations, starting with quantum foundations. A prominent example is QBism, which regards a quantum state as a catalogue of an agent's personal degrees of belief, corresponding to a probability representation $\mathbf{P}_{\mathcal{M}}$~\cite{Fuchs_2013,Appleby2011}. To formulate quantum theory---and specifically the state-update rule---entirely in probabilistic terms, QBism usually considers efficient measurement sets, implemented, for instance, by SIC-POVMs~\cite{Fuchs_2010}. However, \cref{cor:non-redundant} implies that, for systems of unbounded dimension, such probability representations fail to be robust. Our results thus point to a fundamental obstacle to extending the QBist programme to prototypical systems like a harmonic oscillator or the spatial degrees of freedom of a single particle. 

A second example from quantum foundations is the reconstruction program, which aims to retrieve quantum theory from postulates with a clear physical interpretation \cite{Hardy_2001,Wilce2009,Dakic2009,Chiribella2010,Chiribella2011,Hardy2011,Zapo2012,delaTorre2012,Masanes2013,Barnum2014,Masanes2014b,Hardy2015,Muller2016,DAriano_2017,Hohn2017,Hohn2017tool,Hohn_2017rules,Selby2021}. The idea is to consider a broader class of theories---each specified by a state space and an effect space---in which quantum theory appears as a special case. The postulates then correspond to constraints on these spaces. A typical approach is to focus on a small set of measurements, for example local ones~\cite{Pawlowski_2009}, and then constrain the state space by imposing conditions on the probabilities of the measurement outcomes. However, our results show that the quantum state space cannot be robustly retrieved in this manner, unless additional assumptions, such as bounds on the dimension, are imposed.\footnote{One way to avoid this issue is to consider all possible effects rather than a restricted class of measurements. In quantum theory, however, the set of all effects is the dual of the state space, and is therefore just as difficult to characterize as the state space itself.}

Our results also impact quantum information theory. We illustrate this with an example from quantum cryptography, namely generating a key between two parties Alice and Bob, which is secret relative to an adversary Eve \cite{Bennett1984,Ekert1991}. So-called ``post-quantum'' cryptographic schemes achieve this without assuming the validity of quantum theory, relying only on the non-signalling principle~\cite{Ekert1991,Barrett2005}. To analyse such schemes, one represents the measurement choices and outcomes of Alice, Bob and Eve as inputs and outputs to a non-signaling box of the following form\footnote{The non-signalling property ensures that neither party can communicate through their choice of measurement. For example, Alice's marginal distribution $\smash{P_{X|\alpha \gamma \beta}}$ is independent of $\beta$ and $\gamma$.} 
\begin{equation}\label{eq:box}
    \NSbox{\alpha}{\beta}{\gamma}{X}{Y}{Z} \, . 
\end{equation}
The behaviour of this box is fully characterized by the probability distribution~$P_{XYZ|\alpha\beta\gamma}$. An important special case is a non-signaling box that obeys the laws of quantum theory. In this case, the conditional probability distribution~$P_{XYZ|\alpha\beta\gamma}$ corresponds to the probability representation $\mathbf{P}_{\mathcal{M}_{\otimes}}(\rho_{A B E})$ of a quantum state $\rho_{A B E}$ shared by Alice, Bob, and Eve.

\Cref{cor:prod_conv} reveals a problem when one attempts to define the security of a key in this model~\cite{Barrett2005,Masanes2014}: Due to the non-robustness of $\mathbf{P}_{\mathcal{M}_{\otimes}}(\rho_{A B E})$, it is not clear that the security of the box representation implies security according to the standard quantum-cryptographic security definition~\cite{Ferradini_2025}. In other words, post-quantum cryptography is potentially less secure than quantum cryptography.

Quantum field theory (QFT) is another prominent domain in which probability representations are employed. Instead of wave functions, states are typically characterized by their $n$-point correlation functions.\footnote{See also~\cite{Arkanihamed2025} for another discussion of the differences between quantum states and correlation functions.} The injectivity of this representation follows from the Wightman reconstruction theorem \cite{Wightman1956}. Our results, however, suggest that such representations are not robust.

While our treatment assumes that states can be represented as density operators and is therefore not directly applicable to general QFTs, theories satisfying the split property can be approximated by discretized models in which degrees of freedom are localized on a lattice and the state space factorizes accordingly. In this regime, \cref{cor:prod_conv} shows that correlation functions fail to provide a robust state representation. This issue becomes particularly pronounced when considering genuinely non-local observables. A notable example is the set of observables required to test the athermality of Hawking radiation, which manifests in complex non-local correlations. Our results imply that such correlations are not reliably captured by correlation functions.

One might assume that non-robustness is no concern in systems of bounded dimension. There, all norms---and hence their induced topologies---are equivalent. However, non-robust\-ness still appears via a dimen\-sion-dependent scaling factor between the distances~$d_{\mathcal{M}}$ and $\delta$. An application of this phenomenon is data hiding \cite{Terhal2001,DiVincenzo2004,Eggeling2002,Hayden_2004,Koenig2007,Matthews2009,Lancien2013,Chitambar2014,Aubrun2015,Lami2021,Correa2022,Cheng_2023,Mele2025,Ha2025}. Our results can be understood as an extrapolation of this phenomenon to the unbounded-dimensional setting.

In summary, we have argued against probability-based representations: although ubiquitous in physics, they do not share the favourable properties of density operators (see \cref{tab:result}). Probability representations, however, have the advantage of not being tied to the Hilbert space formalism, allowing the exploration of more general theories. This motivates the search for an alternative approach that unifies the theory-independence of probability representations with the robustness of density operators.

\section*{Acknowledgments}

We thank Giulio Chiribella, Thomas Galley, Lucien Hardy, Lluis Masanes, Markus Müller, and Robert Spek\-kens for discussions. We also thank Mile Gu, Blake Stacey, and Matt Weiss for comments on a previous version. 
This work was funded by the Swiss National Science Foundation via project No.\ \mbox{20QU-1\_225171}. We are also grateful for the support from the NCCR SwissMAP, the ETH Zurich Quantum Center, and the Simons Center for Geometry and Physics, where part of this work was carried out. 


\printbibliography

\numberwithin{theorem}{section}
\numberwithin{definition}{section}
\numberwithin{lemma}{section}
\numberwithin{remark}{section}
\numberwithin{proposition}{section}
\numberwithin{corollary}{section}
\numberwithin{equation}{section}

\appendix
\newpage 

\section*{\LARGE{Appendix}}

\section{Topological properties}\label{app:topo}
\begin{definition}\label{def:fragile}
    A set $\mathcal{M}$ of measurements is \emph{fragile} there exist an effect $E$ with finite-dimensional support such that the topologies induced by $d_{\mathcal{M}}$ and $d_{\mathcal{M} \cup \{E, \mathbbm{1}-E\}}$ are different. 
\end{definition}

\begin{remark} \label{rem:fragile}
  Let $\mathcal{N}$ be the set of effects that appear in a measurement set $\mathcal{M}$. Then $\mathcal{M}$ is not fragile if any possible effect with finite support is in the closure of $\mathrm{span}(\mathcal{N})$ with respect to $\|\cdot\|_{\infty}$. This is, in particular, the case for $\mathcal{M}_{\otimes}$.
\end{remark}
\begin{proof}
  Let $E$ be an effect with finite support and $(\rho_n)_{n \in \mathbb{N}}$ a sequence converging to $\rho$ with respect to $d_{\mathcal{M}}$. We first show that, under the assumption made in the remark, the sequence also converges with respect to $d_{\mathcal{M} \cup \{E, \mathbbm{1}-E\}}$. 

 To see this, take a sequence of effects $(E_n)_{n \in \mathbb{N}}$ in the span of the effect of $\mathcal{M}$ that converges to $E$ with respect to $\|\cdot\|_{\infty}$. For every $\varepsilon > 0$ there exists an $E_m = \sum_{i} a^{m}_i N^{m}_i$ with $N_i^m \in \mathcal{N}$ such that $\|E_m -E\|_{\infty} \leq \varepsilon$. So we find 
  \begin{align}
    \begin{split}
      \lim_{n \to \infty} |\tr(E(\rho_n - \rho))| &\leq \lim_{n \to \infty} \sum_i |a^{m}_i| \cdot \bigl|\tr(N_i(\rho_n - \rho))\bigr| + \bigl|\tr((E_m-E)(\rho_n - \rho))\bigr| \\
      &\leq \lim_{n \to \infty}\sum_i |a_i^k| \cdot \|\rho_n - \rho\|_{\mathcal{M}} + 2 \varepsilon \\
      &= 2 \varepsilon.
    \end{split}
  \end{align}
  As this holds for every $\varepsilon > 0$, we have $\lim_{n \to \infty} |\tr(E(\rho_n - \rho))| = 0$. This implies that the sequence $(\rho_n)_{n \in \mathbb{N}}$ converges to $\rho$ also with respect to $d_{\mathcal{M} \cup \{E, \mathbbm{1}-E\}}$. 
  
  We have thus established that every sequence $(\rho_n)_{n \in \mathbb{N}}$ that converges with respect to $d_{\mathcal{M}}$ also converges with respect to $d_{\mathcal{M} \cup \{E, \mathbbm{1}-E\}}$. Because, conversely, $d_{\mathcal{M} \cup \{E, \mathbbm{1}-E\}}$ dominates~$d_{\mathcal{M}}$, the topology induced by $d_{\mathcal{M}}$ is equal to the topology induced by $d_{\mathcal{M} \cup \{E, \mathbbm{1}-E\}}$. The stability condition thus holds. 

  Because product effects with finite support linearly span effects with finite support, the condition is met for $\mathcal{M}_{\otimes}$.
\end{proof}

\proptopo*

\begin{proof}
  We first show that, if the topologies induced by $d_{\mathcal{M}}$ and $\delta$ are identical, then $\mathcal{M}$ is not fragile. If these two topologies are identical, then the topology of $d_{\mathcal{M} \cup \{E, \mathbbm{1}-E\}}$ is finer than the topology of $\delta$. Furthermore, because $d_{\mathcal{N}} \leq \delta$ holds for arbitrary measurement sets $\mathcal{N}$, the topology induced by $d_{\mathcal{M} \cup \{E, \mathbbm{1}-E\}}$ is coarser than that of $\delta$.  Consequently, the topologies induced by $d_{\mathcal{M} \cup \{E, \mathbbm{1}-E\}}$ and $\delta$ are identical. 

  We now show that the stability of $\mathcal{M}$ implies that the topologies induced by $d_{\mathcal{M}}$ and~$\delta$ are identical. Stability means that, for any effect $E$ with finite support, the topologies induced by $d_{\mathcal{M} \cup \{E, \mathbbm{1}-E\}}$ and $d_{\mathcal{M}}$ are identical. Let $\{\ket{n}\}_{n \in \mathbb{N}}$ be an orthonormal basis of $\mathcal{H}$ and $P_N = \sum_{i = 0}^{N} \ketbra{i}$ projectors. For any measurements set $\mathcal{M}$ satisfying the conditions of the proposition and $N \in \mathbb{N}$, the map $\mathsf{P}_N: \rho \mapsto P_N \rho P_N$, where $\rho$ is a state, is continuous with respect to $\| \cdot \|_{\mathcal{M}}$. To see this, let $(\rho_n)_{n \in \mathbb{N}}$ be a sequence that converges to $\rho$ with respect to $\|\cdot\|_{\mathcal{M}}$ and let $\mathcal{E} = \{E_1, \dots E_{N^2}\}$ be a tomographically complete POVM on the support of~$\mathsf{P}_N$. We define the norm $\|\rho \|_{\mathcal{E}} \coloneq \sup_{E \in \mathcal{E}} |\tr(E A)|$ on the image of $\mathsf{P}_N$. Since the topology induced by $d_{\mathcal{M} \cup \{E, \mathbbm{1}-E\}}$ is identical to that of $d_{\mathcal{M}}$ it follows that 
  \begin{align}
    \begin{split}
      \forall E \in \mathcal{E}: \lim_{n \to \infty}|\tr(E P_N (\rho_n-\rho) P_N)| &=  \lim_{n \to \infty} |\tr(E (\rho_n-\rho))| \\ & \leq \lim_{n \to \infty} d_{\mathcal{M} \cup \{E, \mathbbm{1}-E\}}(\rho_n, \rho)  = 0.
    \end{split}
  \end{align}
  As $|\mathcal{E}| < \infty$, we find that $\mathsf{P}_N(\rho_n)$ converges to $\mathsf{P}_N(\rho)$ with respect to $\|\cdot\|_{\mathcal{E}}$, and because the image of $\mathsf{P}_N$ is finite-dimensional also with respect to $\|\cdot\|_{\mathcal{M}}$. As this holds for any converging sequence, we have established that the map $\mathsf{P}_N$ is continuous with respect to~$\| \cdot \|_{\mathcal{M}}$.

  Let $(\rho^{(n)})_{n \in \mathbb{N}}$ be a sequence of states such that $\lim_{n \to \infty} d_{\mathcal{M}}(\rho^{(n)}, \rho) = 0$. Then  
  \begin{align}
    \begin{split}
       0 &= \lim_{n \to \infty} 2 d_{\mathcal{M}}(\rho^{(n)}, \rho) \\
       &= \lim_{N \to \infty} \lim_{n \to \infty} \|\rho^{(n)}- \rho\|_{\mathcal{M}} +  \| \rho - P_N \rho P_N\|_{\mathcal{M}} \\
       &\geq \lim_{N \to \infty} \lim_{n \to \infty} \|\rho^{(n)}- P_N \rho P_N \|_{\mathcal{M}} \\
       &= \lim_{N \to \infty} \lim_{n \to \infty} \| \rho^{(n)}- P_N \rho P_N\|_{\mathcal{M}} +  \|P_N \rho^{(n)} P_N - P_N \rho P_N \|_{\mathcal{M}} \\
       &\geq \lim_{N \to \infty} \lim_{n \to \infty} \|\rho^{(n)} - P_N \rho^{(n)} P_N\|_{\mathcal{M}}\\
       &\geq \lim_{N \to \infty} \lim_{n \to \infty} \sum_{E_i \in M} |\tr(E_i (P_N\rho^{(n)}P_N - \rho^{(n)}))| \\
       &\geq \lim_{N \to \infty} \lim_{n \to \infty} \bigl|1 - \tr(P_N\rho^{(n)}P_N)\bigr| 
      \end{split}
  \end{align}
  where $M \in \mathcal{M}$ is a measurement with effects $\{E_i\}_i$, and we used both the continuity of $\mathsf{P}_N$ with respect to $\|\cdot\|_{\mathcal{M}}$ and that $\lim_{N \to \infty} \delta(P_N \rho P_N, \rho) = 0$. Using this result, we can calculate the limit with respect to $\| \cdot \|_1$,
  \begin{align}
    \begin{split}
      \lim_{n \to \infty} \|\rho^{(n)}- \rho\|_1 &= \lim_{N \to \infty} \lim_{n \to \infty} \bigl|\|\rho^{(n)}- \rho\|_1 - \|P_N \rho P_N - \rho\|_1 \bigr| \\
      &\leq  \lim_{N \to \infty}  \lim_{n \to \infty} \|\rho^{(n)} - P_N  \rho P_N\|_1 \\
      &= \lim_{N \to \infty}  \lim_{n \to \infty} \bigl|\|\rho^{(n)}- P_N\rho P_N\|_1 - \|P_N\rho P_N- P_N \rho^{(n)} P_N\|_1\bigr| \\
      &\leq \lim_{N \to \infty}  \lim_{n \to \infty} \|\rho^{(n)} - P_N \rho^{(n)} P_N\|_1 \\
      &\leq \lim_{N \to \infty}  \lim_{n \to \infty} 2 \sqrt{1-\tr(P_N \rho^{(n)} P_N)} = 0
    \end{split}
  \end{align}
  where in the last inequality we used the gentle measurement lemma \cite[Lemma 9.4.2]{Wilde_2017}. We have thus shown that, if a sequence $(\rho_n)_{n \in \mathbb{N}}$ converges with respect to $\| \cdot \|_{\mathcal{M}}$ then it also converges with respect to $\| \cdot \|_1$. This, together with the fact that $\|\cdot\|_1$ dominates $\|\cdot\|_{\mathcal{M}}$, proves that the topologies induced by $\delta$ and $d_{\mathcal{M}}$ are equal. 
\end{proof}

\begin{remark}\label{rem:topo_counterexample}
  There exists a tomographically complete measurement set $\mathcal{M}$ such that the topologies induced by $d_{\mathcal{M}}$ and $\delta$ are not the same on $\mathcal{D}(\mathcal{H})$.
\end{remark}

\begin{proof}
  Let $\{\ket{n}\}_{n \in \mathbb{N}}$ be a basis and $P = \sum_{n} e^{-n} \ketbra{n}$. We define 
  \begin{equation}
    \mathcal{M} = \bigl\{\{P E P, \mathbbm{1}-P E P\} | \bra{0}E \ket{0} = 0, 0 \leq E \leq \mathbbm{1}\bigr\} \cup \bigl\{\{\mathbbm{1}\}\bigr\}.
  \end{equation}
  It can be readily verified that $\mathcal{M}$ is tomographically complete. To show that the topology induced by $d_{\mathcal{M}}$ is different from the topology induced by $\delta$, we consider the sequence~$(\ketbra{n}{n})_{n \in \mathbb{N}}$. This sequence obviously does not converge with respect to $\delta$. However, it converges with respect to $d_{\mathcal{M}}$. To see the latter, note first that the effect $\mathbbm{1}$ cannot be used to distinguish $\ketbra{0}$ from $\ketbra{n}$. Furthermore, for any effect $E$ with $\bra{0}E\ket{0} = 0$, it holds that $\tr(P E P (\ketbra{n}{n}-\ketbra{0}{0})) = e^{-2n} \bra{n} E \ket{n} \leq e^{-2n}$. 
\end{proof}

 Before proceeding with the proof of our next main statements, we need a technical lemma. The lemma refers to norms that are defined on the space $\mathrm{span}(\mathcal{D}(\mathcal{H}))$, which contains all Hermitian trace-class operators on $\mathcal{H}$. We also remark that any such operator~$A$ can be written as a sum of a positive part $A_+$ and a negative part $A_-$. 

 \begin{lemma} \label{lem:sequenceineq}
 If the norms $\|\cdot\|_1$ and $\|\cdot \|_{\mathcal{M}}$ do not induce the same topology then there exists a sequence $(A^{(n)})_{n \in \mathbb{N}}$ of Hermitian trace-class operators from the set\footnote{The choice of $\frac{1}{11}$ is due to a preference for prime numbers of at least one of the authors.}  
   \begin{equation}
    \mathcal{O}_{\mathrm{diff}} \coloneq \{A:\, {\tr(A) = 0} \land \tr(A_+) \leq {\textstyle \frac{1}{11}}\},
  \end{equation}
  such that 
     \begin{equation} \label{eq:Asequence}
        \lim_{n \to \infty} \|A^{(n)}\|_{\mathcal{M}} = 0,
     \end{equation}   
     whereas there exists $\mu > 0$ such that 
     \begin{equation} \label{eq:Anonconvergence}   
       \forall n \in \mathbb{N}: \quad  \|A^{(n)}\|_{1} \geq \mu.
    \end{equation}
   \end{lemma}

 \begin{proof} 
  The Hilbert space $\mathcal{H}$ is infinite-dimensional, as the norms $\|\cdot\|_1$ and $\|\cdot \|_{\mathcal{M}}$, do not induce the same topology. Furthermore, because $\|\cdot\|_1$ dominates $\|\cdot \|_{\mathcal{M}}$, the inequivalence of norms implies that there exists a sequence $(A^{(n)})_{n \in \mathbb{N}}$ of trace-class operators satisfying~\eqref{eq:Asequence}, while $\lim_{n \to \infty} \|A^{(n)}\|_1 \neq 0$. By restricting to a suitable subsequence, we can also ensure that~\eqref{eq:Anonconvergence} holds. 
  
  In the remainder of the proof, we will show that, by appropriately modifying this sequence, we can ensure that its elements lie in $\mathcal{O}_{\mathrm{diff}}$. 

  We first take care of the condition $\tr(A_+^{(n)}) \leq \frac{1}{11}$. For this, consider the sequence $(B^{(n)})_{n \in \mathbb{N}}$ defined by 
  \begin{align}
    B^{(n)} \coloneq \frac{1}{\max(11 \|A^{(n)}\|_{1}, 1)} A^{(n)}.
  \end{align}
  As $\forall n \in \mathbb{N}: \, \|B^{(n)}\|_{\mathcal{M}} \leq \|A^{(n)}\|_{\mathcal{M}}$, the sequence $(B^{(n)})_{n \in \mathbb{N}}$ still converges with respect to~$\|\cdot \|_{\mathcal{M}}$. Furthermore, for all $n$ it holds that $\|B^{(n)}\|_{1} \geq {\frac{\mu}{11}}$. As $\|\cdot\|_1$ is compatible with data processing, it follows that $(B^{(n)})_{n \in \mathbb{N}}$ has the desired property
    \begin{equation}
      \frac{1}{11} \geq \|B^{(n)}\|_{1} \geq |\tr(\Pi_{\pm} B^{(n)}\Pi_{\pm})| = |\tr(B_{\pm}^{(n)})|
    \end{equation}
    where $\Pi_{\pm}$ is the projector on the positive (negative) part of $B^{(n)}$. This shows that we can, without loss of generality, assume that $\tr(A_+^{(n)}) \leq \frac{1}{11}$.

    Next, we turn to the property $\tr(A^{(n)}) = 0$. Let $\ket{e_0}$ be an eigenvector of $A^{(n)}$ such that the corresponding eigenvalue $\lambda_0$ satisfies $0 \leq \lambda_0 \leq \frac{1}{n}$. Such an eigenvector exists as $\textstyle{\tr(\smash{A_+^{(n)}})} \leq 1$ and the Hilbert space $\mathcal{H}$ is infinite dimensional. Note that 
    the sequence $(B^{(n)})_{n \in \mathbb{N}}$ defined by $B^{(n)} = A^{(n)}- \tr(A^{(n)}) \ketbra{e_0}$ does not converge with respect to $\|\cdot\|_1$ because, for all $n \in \mathbb{N}$,
    \begin{equation}
      \|B^{(n)}\|_1 = \tr(A^{(n)}_+) -  \tr(A^{(n)}_-) -\lambda_0 + |\lambda_0 - \tr(A^{(n)})| \geq \mu - \frac{1}{n}.
    \end{equation}
    However, it still converges with respect to $\|\cdot\|_{\mathcal{M}}$ as
    \begin{equation}
      \|B^{(n)}\|_{\mathcal{M}} \leq \|A^{(n)}\|_{\mathcal{M}} + |\tr(A^{(n)})| \leq  2 \|A^{(n)}\|_{\mathcal{M}} 
    \end{equation}
    where we used that $|\tr(A)| \leq \|A\|_{\mathcal{M}}$ in the last inequality.

 \end{proof}

\topology*

\begin{proof}
    \ref{item:continuous}$\implies$\ref{item:robust}: Let $\Sigma$ be a subset of the state space and $(\rho^{(n)})_{n \in \mathbb{N}}$ a sequence such that $\lim_{n \to \infty} d_{\mathcal{M}}(\rho^{(n)}, \Sigma) = 0$. Then there exists a sequence $(\sigma^{(n)})_{n \in \mathbb{N}} \subset \Sigma$ such that $\lim_{n \to \infty} d_{\mathcal{M}}(\rho^{(n)}, \sigma^{(n)}) = \lim_{n \to \infty} \|\rho^{(n)} - \sigma^{(n)}\|_{\mathcal{M}} = 0$. As the norms induce the same topology, it follows that $\lim_{n \to \infty} \|\rho^{(n)} - \sigma^{(n)}\|_{1} = 0$, which implies $\lim_{n \to \infty} \delta(\rho^{(n)}, \Sigma) = 0$.     
    
    \ref{item:robust}$\implies$\ref{item:continuous}: Assume by contradiction that \ref{item:continuous} does not hold. Then the underlying Hilbert space must be infinite-dimensional, because all norms on a finite-dimensional vector space are equivalent. Let $\{\ket{n}\}_{n \in \mathbb{N}}$ be an orthonormal basis of this Hilbert space and $(A^{(n)})_{n \in \mathbb{N}}$ the sequence of operators in $\mathcal{O}_{\mathrm{diff}}$ as defined by \cref{lem:sequenceineq}. From this, we can build sequences $(\rho^{(n)})_{n \in \mathbb{N}}$ and $(\sigma^{(n)})_{n \in \mathbb{N}}$ of states by
    \begin{align}
        \rho^{(n)} &\coloneq \ketbra{n}{n} \left(1-\tr(A_{+}^{(n)})\right) + A_{+}^{(n)}, \quad \sigma^{(n)} \coloneq \ketbra{n}{n} \left(1+\tr(A_{-}^{(n)})\right) - A_{-}^{(n)}.
    \end{align}
    Note that, because all operators from $\mathcal{O}_{\mathrm{diff}}$ satisfy $\tr(A_+) \leq 1$, these are valid states.    

    
    We calculate the distance between any two states of the two sequences. For $n \neq m$, we have
    \begin{align}
      \begin{split}
        2\delta(\rho^{(n)}, \sigma^{(m)}) &= \|\rho^{(n)} - \sigma^{(m)}\|_{1} \\
        &= \bigl\|\ketbra{n}{n} \left(1-\tr(A_{+}^{(n)})\right) + A_{+}^{(n)} - \ketbra{m}{m} \left(1+\tr(A_{-}^{(m)})\right) + A_{-}^{(m)} \bigr\|_{1} \\
        &\geq \bigl|1-\tr(A_{+}^{(n)})+ \bra{n} A_{+}^{(n)}  \ket{n} +\bra{n} A_{-}^{(m)} \ket{n} \bigr|\\
        &\geq  1-\tr(A_{+}^{(n)}) + \tr(A_{-}^{(m)}) \\
        &\geq 1-\frac{2}{11}
      \end{split}
    \end{align}
    where, in the first inequality, we used that the $1$-norm is non-increasing under any trace non-increasing completely positive map. For $n = m$, we have $\rho^{(n)} - \sigma^{(n)} = A^{(n)}$. Therefore, $\forall n \in \mathbb{N}: \delta(\rho^{(n)}, \sigma^{(n)}) \geq \mu$. Define the set $\Sigma = \{\sigma^{(n)}| n \in \mathbb{N}\}$. By definition, $\lim_{n \to \infty} d_{\mathcal{M}}(\rho^{(n)}, \Sigma) \leq \lim_{n \to \infty} d_{\mathcal{M}}(\rho^{(n)}, \sigma^{(n)}) = 0$. Furthermore, for all $n \in \mathbb{N}$, the bound $\delta(\rho^{(n)}, \Sigma) \geq \min(1-\frac{2}{11}, \mu)$ holds. Therefore, $\mathbf{P}_{\mathcal{M}}$ is not robust.

    \ref{item:continuous} $\implies$ \ref{item:complete}: Assume that the topologies induced by $\| \cdot \|_{\mathcal{M}}$ and $\|\cdot\|_1$ are identical. Then there exists a constant $C > 0$ such that $C\|\cdot\|_1 \leq \|\cdot\|_{\mathcal{M}} \leq \|\cdot \|_1$. Therefore, a sequence is Cauchy or converges with respect to $\| \cdot \|_{\mathcal{M}}$ if and only if it is Cauchy or converges with respect to $\|\cdot\|_{1}$. The space $\mathrm{span}(\mathcal{D}(\mathcal{H}))$ is the set of Hermitian trace-class operators, which is complete with respect to $\|\cdot\|_{1}$ \cite[4.2.2. Corollary]{Murphy1990} and, thus, also with respect to $\| \cdot \|_{\mathcal{M}}$.

    \ref{item:complete} $\implies$ \ref{item:continuous}: Assume the space $\mathrm{span}(\mathcal{D}(\mathcal{H}))$ with the topology induced by$\| \cdot \|_{\mathcal{M}}$ is complete and, thus, a Banach space. Consider the identity map $(\mathrm{span}(\mathcal{D}(\mathcal{H})),\|\cdot\|_1) \to (\mathrm{span}(\mathcal{D}(\mathcal{H})),\|\cdot\|_{\mathcal{M}})$. This map is continuous as $\|\cdot\|_{\mathcal{M}} \leq \|\cdot\|_{1}$. By the open mapping theorem for continuous linear functions on Banach spaces \cite[Theorem III.11]{Reed1972}, the inverse of this map is also continuous. Thus, the topologies induced by $\| \cdot \|_{\mathcal{M}}$ and $\|\cdot\|_1$ are identical. 
\end{proof}

\begin{remark}
  The representation $\mathbf{P}_{\mathcal{M}}$ being robust is also equivalent to $\mathrm{span}(\mathcal{D}(\mathcal{H}))$ not being meagre\footnote{A topological space $X$ is meagre if it is the countable union of nowhere dense sets, i.e., sets whose closures have an empty interior.} with respect to the topology induced by $\| \cdot \|_{\mathcal{M}}$.
\end{remark}

\begin{proof}
  If the topologies induced by $\| \cdot \|_{\mathcal{M}}$ and $\| \cdot \|_{1}$ are identical then, by \cref{prop:topology}, $\mathrm{span}(\mathcal{D}(\mathcal{H}))$ equipped with $\| \cdot \|_{\mathcal{M}}$ is a Banach space. It follows directly from the Baire category theorem \cite[Theorem III.8]{Reed1972} that a Banach space is not meagre. 

  Conversely, assume the space $\mathrm{span}(\mathcal{D}(\mathcal{H}))$ with the topology induced by $\| \cdot \|_{\mathcal{M}}$ is not meagre. Consider the family of functionals $\mathcal{F} \coloneq \{A \in \mathrm{span}(\mathcal{D}(\mathcal{H})) \mapsto \tr(E A)| 0 \leq E \leq \mathbbm{1}\}$. This family of functionals is pointwise bounded on $\mathrm{span}(\mathcal{D}(\mathcal{H}))$: 
  \begin{equation}
    \forall A \in \mathrm{span}(\mathcal{D}(\mathcal{H})): \ \sup_{0 \leq E \leq \mathbbm{1}} \tr(EA) = \|A\|_+ < \infty,
  \end{equation}
  where $\|A\|_+ = \frac{1}{2} |\tr(A)| +  \frac{1}{2} \|A\|_1 \leq  \|A\|_1$ is the generalized trace distance. Note that this norm is equivalent to $\| \cdot \|_1$.
  As $\mathrm{span}(\mathcal{D}(\mathcal{H}))$ is not meagre with respect to the topology induced by $\|\cdot\|_{\mathcal{M}}$, we can apply Banach-Steinhaus \cite[Theorem 2.5]{Rudin1991}, which implies that the family of functionals $\mathcal{F}$ is pointwise equicontinous with respect to~$\|\cdot\|_{\mathcal{M}}$. Therefore, for every $\varepsilon > 0$, there is a $\delta > 0$ such that $\|A\|_{\mathcal{M}} \leq \delta \implies \|A\|_{+} \leq \varepsilon$. Thus, any sequence that converges with respect to~$\|\cdot\|_{\mathcal{M}}$ also converges with respect to~$\|\cdot\|_{+}$, and due to the equivalence of $\|\cdot\|_{+}$ with $\| \cdot \|_1$ also with respect to~$\| \cdot \|_1$. The reverse is also true since $\|\cdot\|_{\mathcal{M}} \leq \|\cdot\|_{1}$. Thus, the topologies induced by~$\| \cdot \|_{\mathcal{M}}$ and $\|\cdot\|_1$ are identical. 
\end{proof}

\section{Asymptotic entropy}\label{app:entropy}
Here we give the formal version of \cref{def:epsilon_decoded}. First, we define the allowed decoding operations. 

\begin{definition}
  Let $\mathcal{O}$ be a collection of positive operators and $M$ a POVM on a measurable space $(\mathbb{O},\Sigma_M)$. A \emph{decoding operation} $\mathcal{D}_{M}$ for $P_M$ maps lists of probabilities $(P_{E})_{E \in \mathcal{O}}$ to a function
  \begin{equation}
    \mathcal{D}_{M}((P_{E})_{E \in \mathcal{O}}): F \in \mathcal{S}\mapsto \sum_{E} p_E(F) P_{E} \in \mathbb{R}
  \end{equation}
  where the sum runs over finitely many elements of $\mathcal{O}$, $\{p_E\}_{E \in \mathcal{O}} \subset \mathbb{R}_{+}$, and $\mathcal{S}$ is a semi-algebra that generates $\Sigma_M$ and is closed under finite disjoint unions.
\end{definition}

To define $\varepsilon$-decoding, we use the following norm on functions $P: \mathcal{S} \to \mathbb{R}$:
\begin{equation}
  \|P\| \coloneq \sup_{F \in \mathcal{S}} |P(F)|.
\end{equation}

\begin{definition}\label{def:decoded_technical}
    Let $\mathcal{O}$ be a collection of positive operators. We say that $\mathbf{P}_\mathcal{M}$ can be \emph{$\varepsilon$-decoded} from $\mathcal{O}$ if for every $M \in \mathcal{M}$ there exists a sequence of decoding operations $(\mathcal{D}^{(k)}_{M})_{k \in \mathbb{N}}$ such that 
    \begin{equation}\label{eq:convergence}
     \forall \rho \in \mathcal{D}(\mathcal{H}): \quad \limsup_{k \to \infty} \sup_{M \in \mathcal{M}} \left\|\mathcal{D}^{(k)}_{M}\bigl(\bigl(\tr(E \rho)\bigr)_{E \in \mathcal{O}}\bigr) - P_M(\rho) \right\| < \varepsilon.
    \end{equation}
\end{definition}

\begin{lemma}\label{lem:recovery}
   Let $\mathbf{P}_\mathcal{M}$ be $\varepsilon$-decodable from a set of positive operators $\mathcal{O}$ and let $\rho$ and $\sigma$ be states. Then for every $M \in \mathcal{M}$ there exists a semi-algebra~$\mathcal{S}_M$ generating $\Sigma_M$ such that for every $F\in\mathcal{S}_M$ there exists an operator $N^{M}_F$ in the positive span of $\mathcal{O}$ satisfying
  \begin{align}
    \tr(N_{F}^{M} \rho) \leq 1+\varepsilon, \quad \tr(N_{F}^{M} \sigma) &\leq 1+\varepsilon, \\ 
    \left||\tr(M(F) (\rho -\sigma))| -  |\tr(N^{M}_F(\rho-\sigma))|\right| &\leq 2\epsilon.
  \end{align}
\end{lemma}
\begin{proof}
  Let $(\mathcal{D}^{(n)}_{M})_{n \in \mathbb{N}, M \in \mathcal{M}}$ be the decoding operations that exist because $\mathbf{P}_\mathcal{M}$ can be $\varepsilon$-decoded from~$\mathcal{O}$, and let $k \in \mathbb{N}$ be such that
  \begin{equation}
    \sup_{M \in \mathcal{M}} \left\|\mathcal{D}^{(k)}_{M}\bigl(\bigl(\tr(E \rho)\bigr)_{E \in \mathcal{O}}\bigr) - P_M(\rho)\right\| \leq \varepsilon \text{ and }\sup_{M \in \mathcal{M}} \left\|\mathcal{D}^{(k)}_{M}\bigl(\bigl(\tr(E \sigma)\bigr)_{E \in \mathcal{O}}\bigr) - P_M(\sigma) \right\| \leq \varepsilon.
  \end{equation}
  Then, for every $M \in \mathcal{M}$, consider the decoding operation $\mathcal{D}^{(k)}_{M}$. Let $\mathcal{S}_{M}$ be the semi-algebra specified by the decoding operation $\mathcal{D}^{(k)}_{M}$, and let $N_F^{M} = \sum_{E \in \mathcal{O}} p_E^{M}(F) E$, where $p_E^{M}(F)$ is also specified by the decoding operation.  Then, by \cref{def:decoded_technical}, it follows that for all $M \in \mathcal{M}$ and $F \in \mathcal{S}_{M}$
  \begin{align}
    \begin{split}
        \left||\tr(P_M(F) (\rho -\sigma))| -  |\tr(N_{F}^{M}(\rho-\sigma))|\right| &\leq |\tr(P_M(F)(\rho-\sigma)) - \tr(N_{F}^{M} (\rho-\sigma))| \\
        &\leq  |\tr(P_M(F)\rho)-\tr(N_{F}^{M} \rho)| \\
        & \qquad + |\tr(P_M(F)\sigma)-\tr(N_{F}^{M}  \sigma)| \\
        &\leq 2 \varepsilon.
    \end{split}
  \end{align}
  Furthermore, \cref{eq:convergence} ensures that for all $M \in \mathcal{M}$ and $F \in \mathcal{S}_{M}$
  \begin{equation}
    \tr(N_{F}^{M} \rho) \leq \tr(P_M(F) \rho) + \varepsilon \leq 1 + \varepsilon.
  \end{equation}
  The same argument also applies to $\sigma$.
\end{proof}

\subsection{Proof of \texorpdfstring{\cref{thm:entropy}}{}}
Let $\mathcal{H}$ be a $d$-dimensional Hilbert space, $\ket{\psi} \in \mathcal{H}$ and $\rho \in \mathcal{D}(\mathcal{H})$. For the proof of \cref{thm:entropy}, the following function is useful
\begin{align}
    \begin{split}\label{eq:def_D_psi}
      D^{\rho}_{\ket{\psi}}: U(\mathcal{H}) &\to \mathbb{R} \\
      U &\mapsto \bigl|\bra{\psi}U\rho U^{\dagger}\ket{\psi}-\frac{1}{d}\bigr|.
    \end{split}
\end{align}
Let us now prove some properties of this function.

\begin{lemma}\label{lem:lip_const_min}
  For every $\ket{\psi}\in\mathcal{H}$, the function $D^{\rho}_{\ket{\psi}}$ is $2^{-H_{\min}(\rho)+\frac{3}{2}}$-Lipschitz with respect to the $2$-norm on the group $U(\mathcal{H})$ of unitaries on $\mathcal{H}$.
\end{lemma}

\begin{proof}
  First note that it suffices to show for every $\ket{\psi}$ that 
  \begin{equation}
    |{D^{\rho}_{\ket{\psi}}(U) -D^{\rho}_{\ket{\psi}}(\mathbbm{1})}| \leq 2^{-H_{\min}(\rho)+\frac{3}{2}} \|{U - \mathbbm{1}}\|_2,
  \end{equation}
  as then the lemma follows from 
  \begin{equation}
   |D^{\rho}_{\ket{\psi}}(U) -D^{\rho}_{\ket{\psi}}(U')| = |D^{\rho}_{(U')^{\dagger}\ket{\psi}}((U')^{\dagger}U) -D^{\rho}_{(U')^{\dagger} \ket{\psi}}(\mathbbm{1})|,
  \end{equation}
  and the fact that, for all $X \in \mathrm{Herm}(L)$ and $U \in U(L)$ it holds that $\|UX\|_2 = \|X\|_2$.

  We define $\theta \in [0, \pi)$ by
  \begin{equation}
    \cos(\theta) = |\bra{\psi}U\ket{\psi}|
  \end{equation}
  and a state $\ket*{\bar{\psi}}$ such that $\braket*{\bar{\psi}}{\psi} = 0$, as well as
  \begin{equation}
    U^{\dagger} \ket{\psi} = e^{i\varphi}(\cos(\theta) \ket{\psi} + \sin(\theta) \ket*{\bar{\psi}}).
  \end{equation}
  Then, we find that 
  \begin{align}
    \begin{split}
      |D^{\rho}_{\ket{\psi}}(U) -D^{\rho}_{\ket{\psi}}(\mathbbm{1})| &\leq |\bra{\psi}\rho\ket{\psi}-\bra{\psi}U\rho \, U^{\dagger}\ket{\psi}| \\
      &=  |(1-\cos(\theta)^2)\bra{\psi}\rho\ket{\psi} - \sin(\theta)^2 \bra*{\bar{\psi}}\rho\ket*{\bar{\psi}} - 2 \sin(\theta) \cos(\theta)\mathrm{Re}(\bra*{\bar{\psi}}\rho\ket{\psi}) | \\
      &\leq   \sin(\theta)^2(\bra{\psi}\rho\ket{\psi} + \bra*{\bar{\psi}}\rho\ket*{\bar{\psi}}) + |2 \sin(\theta) \cos(\theta)| |\mathrm{Re}(\bra*{\bar{\psi}}\rho\ket{\psi}) |
    \end{split}
  \end{align}

  We estimate $\bra{\psi}\rho\ket{\psi}$, $\bra{\bar{\psi}}\rho\ket{\bar{\psi}}$ and $|\mathrm{Re}(\bra*{\bar{\psi}}\rho\ket{\psi}) |$. By definition of the min-entropy, it holds that $\rho \leq 2^{-H_{\min}(\rho)}\mathbbm{1}$. Thus, we find that $\bra{\psi}\rho\ket{\psi} \leq 2^{-H_{\min}(\rho)}$ and  $\bra{\bar{\psi}}\rho\ket{\bar{\psi}} \leq 2^{-H_{\min}(\rho)}$. To bound $|\mathrm{Re}(\bra*{\bar{\psi}}\rho\ket{\psi})|$, we observe that 
  \begin{align}
    \begin{split}
      |\mathrm{Re}(\bra*{\bar{\psi}}\rho\ket{\psi})|^2  &\leq |\bra*{\bar{\psi}}\rho\ket{\psi}|^2 \\
      &= \bra*{\bar{\psi}}\rho\ketbra{\psi}\rho \ket{\bar{\psi}} \\
      &\leq \bra*{\bar{\psi}}\rho^2 \ket{\bar{\psi}} \\
      &\leq  2^{-2 H_{\min}(\rho)}.
    \end{split}
  \end{align}
  Thus, also $|\mathrm{Re}(\bra*{\bar{\psi}}\rho\ket{\psi})| \leq 2^{-H_{\min}(\rho)}$. Plugging this result into the previous calculation, we find that
  \begin{align}
    \begin{split}
    |D^{\rho}_{\ket{\psi}}(U) -D^{\rho}_{\ket{\psi}}(\mathbbm{1})| &\leq 2^{-H_{\min}(\rho)+1} \left(\sin(\theta)^2+ |\cos(\theta) \sin(\theta)|\right) \\
    &= 2^{-H_{\min}(\rho)+1} \sqrt{\sin(\theta)^4 + \cos(\theta)^2 \sin(\theta)^2 + 2\sin(\theta)^2|\cos(\theta) \sin(\theta)| } \\
    &\leq 2^{-H_{\min}(\rho)+1} \sqrt{\sin(\theta)^4 + \cos(\theta)^2 \sin(\theta)^2 + \sin(\theta)^2} \\
    &=2^{-H_{\min}(\rho)+\frac{3}{2}} \sqrt{\sin(\theta)^2} = 2^{-H_{\min}(\rho)+\frac{3}{2}} \sqrt{1-\cos(\theta)^2} \\
    \end{split}
  \end{align}

  Let us now relate this result to the $2$-norm. First, we observe that
  \begin{align}
    \begin{split}
      \|\ket{\psi}-U\ket{\psi}\|^2 &= 2- 2\, \mathrm{Re}(\bra{\psi} U \ket{\psi}) \\
      &= 2- 2\cos(\varphi)\cos(\theta) -(1-\cos(\theta)^2) +(1-\cos(\theta)^2) \\
      &\geq \cos(\varphi)^2- 2\cos(\varphi)\cos(\theta) + \cos(\theta)^2 +(1-\cos(\theta)^2) \\
      &= (\cos(\varphi)-\cos(\theta))^2 +(1-\cos(\theta)^2) \\
      &\geq 1-\cos(\theta)^2 \ .
    \end{split}
  \end{align}
  Furthermore, $\|\ket{\psi}-U\ket{\psi}\|^2$ can be bounded by $\|\mathbbm{1}-U\|_2^2$
  \begin{align}
    \begin{split}
      \|\ket{\psi}-U\ket{\psi}\|^2 &= 2- 2\, \mathrm{Re}(\bra{\psi} U \ket{\psi}) \\
      &= \tr(\ketbra{\psi} (2\mathbbm{1}-U - U^{\dagger}))\\
      &= \tr(\ketbra{\psi} (\mathbbm{1}-U)^{\dagger}(\mathbbm{1}-U))\\
      &\leq \tr((\mathbbm{1}-U)^{\dagger}(\mathbbm{1}-U)) = \|\mathbbm{1}-U\|^2_2
    \end{split}
  \end{align}
  where in the last inequality we used that $(\mathbbm{1}-U)^{\dagger}(\mathbbm{1}-U)$ is a positive operator.
  Thus, 
  \begin{equation}
    |D^{\rho}_{\ket{\psi}}(U) -D^{\rho}_{\ket{\psi}}(\mathbbm{1})| \leq 2^{-H_{\min}(\rho)+\frac{3}{2}} \|\ket{\psi}-U\ket{\psi}\| \leq 2^{-H_{\min}(\rho)+\frac{3}{2}} \|\mathbbm{1}-U\|^2_2.
  \end{equation}
\end{proof}

\begin{lemma}\label{lem:average}
  The average of $D^{\rho}_{\ket{\psi}}(U)$ over $U$ chosen according to the Haar measure satisfies
  \begin{equation}
    \langle D^{\rho}_{\ket{\psi}} \rangle = \int D^{\rho}_{\ket{\psi}}(U) \, dU\leq 2^{- \log(d) - \frac{1}{2} H_{\min}(\rho)}.
  \end{equation}
\end{lemma}

\begin{proof}
  To calculate this average, we use \cite[Theorem 3.3]{Dupuis_2014}, from which it follows that
  \begin{equation}
    \int |\bra{\psi}U\rho U^{\dagger}\ket{\psi}-\frac{1}{d}| \, dU \leq 2^{-\frac{1}{2}(H_2(\rho)+H_2(\tau))}
  \end{equation}
  where $H_2(\rho) \coloneq - \log(\tr(\rho^2)) \leq H_{\min}(\rho)$ and $\tau = \frac{1}{d}\ketbra{\psi}$. Thus, we find 
  \begin{equation}
    \int |\bra{\psi}U\rho U^{\dagger}\ket{\psi}-\frac{1}{d}| \, dU \leq 2^{-\log(d)-\frac{1}{2}H_{\min}(\rho)} .
  \end{equation}
\end{proof}

\begin{lemma}\label{lem:airplane}
  Let $\mathcal{H}$ be an infinite-dimensional Hilbert space, $\{\ket*{n}\}_{n \in \mathbb{N}}$ an orthonormal basis of $\mathcal{H}$, and $(\rho^{(m)})_{m \in \mathbb{N}}, (\sigma^{(n)})_{n \in \mathbb{N}}$ two sequences of states defined by
  \begin{equation}
    \rho^{(m)} \coloneq \frac{1}{m} \sum_{k = 0}^{m-1} \ketbra{k}{k}, \, \sigma^{(n)} = \sum_{k = 0}^{n-1} \frac{2(n-k)}{n (n+1)} \ketbra{k},
  \end{equation}
  then there is an $N \in \mathbb{N}$ such for all $\forall m \in \mathbb{N}, n \geq N: \, \|\mathrm{spec}(\sigma^{(n)}) - \mathrm{spec}(\rho^{(m)})\|_1 \geq \frac{2}{11}$.
\end{lemma}

\begin{proof}
  We distinguish different cases.

  \noindent
  \textbf{Case $m \geq n-1$:} In this case, we can write
  \begin{align}
      \|\mathrm{spec}(\sigma^{(n)}) - \mathrm{spec}(\rho^{(m)})\|_1  &= \left(\sum_{k = 0}^{n-1} \left|\frac{2 (n-k)}{n(n+1)} - \frac{1}{m}\right| + \left|0-(m-n)\frac{1}{m}\right| \right).
  \end{align}
  The sum can be divided into two parts of equal magnitude: the part where the terms in the absolute value are positive and one where they are negative. The former is the case when $k \in \{0, \dots, k_{\max}\}$ with $k_{\max} \coloneq n - \frac{n (n+1)}{2 m}$. Therefore, we find by a straightforward but tedious calculation\ifthenelse{\boolean{showcalc}}{}{\footnote{To see the tedious calculations in this proof, download the source code and enable the option ``showcalc''.}}
  \ifthenelse{\boolean{showcalc}}{
      \begin{align}
        \begin{split}
          \frac{1}{2}\|\mathrm{spec}(\sigma^{(n)}) - \mathrm{spec}(\rho^{(m)})\|_1 &= \sum_{k = 0}^{k_{\max}} \left|\frac{2 (n-k)}{n(n+1)} - \frac{1}{m}\right| \\
          &= (k_{\max} +1)\left(\frac{2}{(n+1)} - \frac{1}{m}\right) -\frac{2}{n(n+1)} \frac{k_{\max}(k_{\max} +1)}{2}\\
          &= \left(1 + n - \frac{n(n+1)}{2m}\right)\left(\frac{2}{n+1} - \frac{1}{m}\right) \\ 
          &\qquad - \frac{1}{n(n+1)} \left(n - \frac{n(n+1)}{2m}\right)\left(n - \frac{n(n+1)}{2m} +1\right) \\
          &= (n+1)\left(1 - \frac{n}{2m}\right)\left(\frac{2}{n+1} - \frac{1}{m}\right) - \left(1 - \frac{n}{2m}\right)\left(1 - \frac{n+1}{2m}\right) \\
          &= \left(1 - \frac{n}{2m}\right)\left(1 - \frac{n+1}{2m}\right) \\
          &\geq \frac{1}{4} + O\left(\frac{1}{n}\right).
        \end{split}
    \end{align}  
  }{
      \begin{align}
        \begin{split}
          \frac{1}{2}\|\mathrm{spec}(\sigma^{(n)}) - \mathrm{spec}(\rho^{(m)})\|_1  &= \sum_{k = 0}^{k_{\max}} \left|\frac{2 (n-k)}{n(n+1)} - \frac{1}{m}\right| + \left|0-(m-n)\frac{1}{m}\right| \\
          &\geq \frac{1}{4} + O\left(\frac{1}{n}\right).
        \end{split}
      \end{align} 
  }

  \noindent
  \textbf{Case $m < n-1$:} 
  In this case, we can write 
  \begin{align}
   \|\mathrm{spec}(\sigma^{(n)}) - \mathrm{spec}(\rho^{(m)})\|_1  &= \frac{1}{2} \sum_{k = 0}^{m} \left|\frac{2 (n-k)}{n(n+1)} - \frac{1}{m}\right| + \frac{1}{2} \sum_{k = m+1}^{n-1} \left|\frac{2 (n-k)}{n(n+1)} - 0\right|. 
  \end{align}
  We consider two subcases. 

  \textbf{Case $k_{\max} \geq 0$:}
  In this case $n - \frac{n(n+1)}{2m} \geq 0 \iff m \geq \frac{n+1}{2}$. As before, we divide the sum into two parts of equal magnitude: the part where the terms in the absolute value are positive and where they are negative. The latter is the case if $m \geq k \geq k_{\max}$.  Thus, we find by a straightforward but tedious calculation\ifthenelse{\boolean{showcalc}}{, where we indicate by \textcolor{Blue}{blue} a term that will be absorbed into $O\left(\frac{1}{n}\right)$ and by \textcolor{BrickRed}{red} a term that will contribute to the next inequality:}{:}
  \ifthenelse{\boolean{showcalc}}{
    \begin{align}
      \begin{split}\label{eq:long_calc}
        \frac{1}{2}\|\mathrm{spec}(\sigma^{(n)}) - \mathrm{spec}(\rho^{(m)})\|_1  &= \sum_{k = k_{\max}}^{m} \left|\frac{2 (n-k)}{n(n+1)} - \frac{1}{m}\right| \\
        &= (m-k_{\max} +1)\left(\frac{1}{m} - \frac{2}{n+1}\right) + \frac{1}{n(n+1)}\textcolor{BrickRed}{\left(m(m+1) - k_{\max}(k_{\max} +1)\right)} 
        \\
        &\,\textcolor{BrickRed}{\geq} \, (m-k_{\max} +1)\left(\frac{1}{m} - \frac{2}{n+1}\right) + \frac{1}{n(n+1)}\left(m^2 - (k_{\max} +1)^2\right) 
        \\
        &= \left(m- n + \frac{n (n+1)}{2 m} +1\right)\left(\frac{1}{m} - \frac{2}{n+1}\right)  \\
        &\qquad + \frac{1}{n(n+1)}\left(m^2 - \left( n+1 - \frac{n (n+1)}{2 m}\right)^2\right) 
        \\
        &= 1-\frac{2m}{n+1} -\frac{n}{m} +\frac{2n}{n\textcolor{Blue}{+1}} + \frac{n(n+1)}{2m^2} -\frac{n}{m} 
        + \frac{1}{m} \textcolor{Blue}{ - \frac{2}{n+1}} \\
        &\qquad +  \frac{m^2}{n(n\textcolor{Blue}{+1})} - \frac{n+1}{n}\left(1- \frac{n}{2m}\right)^2 \\
        &= 1-\textcolor{BrickRed}{\frac{2m}{n+1}} -\frac{n}{m} + 2 + \frac{n(n+1)}{2m^2} -\frac{n}{m} 
        + \frac{1}{m}  \\
        &\qquad + \frac{m^2}{n^2}  - \frac{n+1}{n}\left(1- \frac{n}{2m}\right)^2 + \textcolor{Blue}{ O\left(\frac{1}{n}\right)}
        \\
        &\,\textcolor{BrickRed}{\geq}\, 1-\frac{2m}{n} -\frac{2n}{m} +2 + \frac{n(n+1)}{2m^2}
        + \frac{1}{m} + \frac{m^2}{n^2} - \frac{n+1}{n} \left(1- \frac{n}{2m}\right)^2  + \textcolor{Blue}{ O\left(\frac{1}{n}\right)}
        \\
        &= 1-\frac{2m}{n} + \frac{m^2}{n^2} +2 -\frac{2n}{m}  + \frac{n(n+1)}{2m^2}
        + \frac{1}{m} - \frac{n+1}{n}\left(1- \frac{n}{2m}\right)^2  + \textcolor{Blue}{ O\left(\frac{1}{n}\right)}
        \\
        &= \left(1-\frac{m}{n}\right)^2 + 2  -\frac{2n}{m} + \frac{n^2}{2m^2} + \frac{n}{2m^2}
        + \frac{1}{m} - \frac{n+1}{n} \left(1- \frac{n}{2m}\right)^2  + \textcolor{Blue}{ O\left(\frac{1}{n}\right)}
        \\
        &= \left(1-\frac{m}{n}\right)^2 + 2\left(1- \frac{n}{2m}\right)^2 + \frac{n}{2m^2}
        + \frac{1}{m} - \frac{n+1}{n} \left(1- \frac{n}{2m}\right)^2  + \textcolor{Blue}{ O\left(\frac{1}{n}\right)}
        \\
        &= \left(1-\frac{m}{n}\right)^2 +\left(1- \frac{n}{2m}\right)^2 +  \textcolor{Blue}{\frac{n}{2m^2}
        + \frac{1}{m} - \frac{1}{n} \left(1- \frac{n}{2m}\right)^2}  + \textcolor{Blue}{ O\left(\frac{1}{n}\right)}
        \\
        &= \left(1-\frac{m}{n}\right)^2 + \left(1- \frac{n}{2m}\right)^2  + \textcolor{Blue}{ O\left(\frac{1}{n}\right)}.
      \end{split}
    \end{align}
    
  }{
    \begin{align}
      \begin{split}\label{eq:long_calc}
        \frac{1}{2}\|\mathrm{spec}(\sigma^{(n)}) - \mathrm{spec}(\rho^{(m)})\|_1  &= \sum_{k = k_{\max}}^{m} \left|\frac{2 (n-k)}{n(n+1)} - \frac{1}{m}\right| \\
        &\geq \left(1-\frac{m}{n}\right)^2 + \left(1- \frac{n}{2m}\right)^2  + O\left(\frac{1}{n}\right).
      \end{split}
    \end{align}
  }

  Consider the function $f(x) \coloneq  (1-\frac{1}{x})^2 + (1- \frac{x}{2})^2$. To find the minimum of this function, we set its derivative to zero
  \begin{equation}
    2\frac{1}{x^2}\Big(1-\frac{1}{x}\Big) - \Big(1- \frac{x}{2}\Big) = 0.
  \end{equation}
  This equation is solved by $x = \sqrt{2}$. Therefore, the minimum of this function is $f(\sqrt{2}) = 2(1-\frac{1}{\sqrt{2}})^2 \approx 0.17 > \frac{1}{11}$, which puts a lower bound on \cref{eq:long_calc}.

  \textbf{Case $k_{\max} < 0$:}
  In this case $n - \frac{n(n+1)}{2m} < 0 \iff m < \frac{n+1}{2}$ and we find by another tedious calculation
  \ifthenelse{\boolean{showcalc}}{
    \begin{align} \nonumber
         \frac{1}{2} \|\mathrm{spec}(\sigma^{(n)}) - \mathrm{spec}(\rho^{(m)})\|_1  &= \sum_{k = 0}^{m} \left|\frac{2 (n-k)}{n(n+1)} - \frac{1}{m}\right| \\ \nonumber
        &= (m +1)\left(\frac{1}{m} - \frac{2}{n+1}\right) + \frac{m(m+1)}{n(n+1)}
        \\ \nonumber
        &= \frac{m+1}{m} - \frac{2(m+1)}{n+1} + \frac{m(m+1)}{n(n+1)}
        \\\nonumber
        &= \frac{(m+1)}{m} + \frac{m(m+1)-2(m+1)n}{n(n+1)} \\
        &= \frac{(m+1)}{m} + \frac{m^2+m-2mn-2n}{n(n+1)} \\ \nonumber
        &\geq \frac{(m+1)}{m} - \frac{m(2n-m)}{n^2} + O\left(\frac{1}{n}\right)\\\nonumber
        &= \frac{(m+1)n^2 - m^2(2n-m)}{m n^2} + O\left(\frac{1}{n}\right)\\\nonumber
        &= \frac{mn^2+n^2 - 2n m^2 + m^3}{m n^2} + O\left(\frac{1}{n}\right)\\\nonumber
        &= \frac{m(n-m)^2 +n^2}{m n^2} + O\left(\frac{1}{n}\right)\\\nonumber
        &= \frac{(n-m)^2}{n^2} + \frac{1}{m} + O\left(\frac{1}{n}\right)\\\nonumber
        &= \left(1-\frac{m}{n}\right)^2 + \frac{1}{m} + O\left(\frac{1}{n}\right)\\\nonumber
        &\geq \frac{1}{4} + O\left(\frac{1}{n}\right)
    \end{align}
  }{
    \begin{align}
      \begin{split}
        \frac{1}{2}\|\mathrm{spec}(\sigma^{(n)}) - \mathrm{spec}(\rho^{(m)})\|_1  &= \sum_{k = 0}^{m} \left|\frac{2 (n-k)}{n(n+1)} - \frac{1}{m}\right| \\ 
        &\geq \frac{1}{4} + O\left(\frac{1}{n}\right)
      \end{split}
    \end{align}
  }   
\end{proof}

\minentropy*

\begin{proof}
  Let $\{\Pi_{2^n}\}_{n \in \mathbb{N}}$ be the family of projectors, $\{\mathcal{O}_n\}_{n \in \mathbb{N}}$ the sequence of sets of effects, and $(\eta_n)_{n \in \mathbb{N}}$ the zero-sequence that yield the asymptotic entropy at most~$(\varepsilon_n 2^n)_{n \in \mathbb{N}}$. The family of projectors $\{\Pi_{2^n}\}_{n \in \mathbb{N}}$ commutes. Therefore, there exists a basis $\{\ket*{n}\}_{n \in \mathbb{N}}$ such that the span of the first $2^n$ basis elements is contained in the support of $\Pi_{2^n}$. Using this basis we define the projectors $P_{n} = \sum_{i = 0}^{n-1} \ketbra{i}$, the sequence
  \begin{equation}
    \sigma^{(n)} \coloneq \sum_{k = 0}^{2^n-1} \frac{2(2^n-k)}{2^n (2^n+1)} \ketbra{k},
  \end{equation}
  and the set
  \begin{equation}
    \Sigma \coloneq \bigcup_{n = 0}^{\infty} \mathcal{B}^{\frac{1}{13}}\left(\frac{P_n}{n} \right),
  \end{equation}
  where $\mathcal{B}^{\varepsilon}(\rho) \coloneq \{\sigma \in \mathcal{D}(\mathcal{H})| \delta(\sigma,\rho) < \varepsilon \}$. Furthermore, we denote by $\mathcal{H}_n$ the Hilbert space spanned by $\{\ket{0}, \dots, \ket{n}\}$.
  
  Next, we show, that there is an $N \in \mathbb{N}$ such that for all $n > N$ and unitaries $U$ on $\mathcal{H}_{2^n}$ it holds that $U \sigma^{(n)} U^{\dagger} \notin \Sigma$. To show this, we use \cite[Lemma IV.3.1]{Bhatia1997} 
  \begin{equation}
    \forall \rho,\sigma: \, \|\rho - \sigma\|_1 \geq \|\mathrm{spec}(\sigma) - \mathrm{spec}(\rho)\|_1
  \end{equation}
  where $\mathrm{spec}(\rho)$ is the probability distribution defined by the ordered list of eigenvalues of~$\rho$. As the spectrum of an operator is invariant under unitary operations, it follows from \cref{lem:airplane} that there is an $N \in \mathbb{N}$ such that $\forall U \in U(\mathcal{H}_{2^n}),  n > N: \, \delta(U \sigma^{(n)} U^{\dagger}, \Sigma) \geq \frac{1}{11}$. Therefore, $\forall U \in U(\mathcal{H}_{2^n}): \, \lim_{n \to \infty} \delta(U \sigma^{(n)} U^{\dagger},\Sigma) \neq 0$.

  Our goal is now to show that there exists a sequence of unitaries $\{U^{(n)}\}_{n \in \mathbb{N}}$ such that~$U^{(n)}$ has support on $\mathcal{H}_{2^n}$ and $\lim_{n \to \infty} d_{\mathcal{M}}(U^{(n)}\sigma^{(n)} (U^{(n)})^{\dagger},  2^{-n} P_{2^n})  = 0$. If such a sequence of unitaries exists, this implies  that $\lim_{n \to \infty} d_{\mathcal{M}}(U^{(n)}\sigma^{(n)} (U^{(n)})^{\dagger}, \Sigma)  = 0$ and, by the above result about $\delta$, the representation $\mathbf{P}_{\mathcal{M}}$ is not robust. 
  
  Because the support of $P_{2^n}$ and $\sigma^{(n)}$ is contained in the support of $\Pi_{2^n}$, for any measurement $M \in \mathcal{M}$ and unitary $U \in U(\mathcal{H}_{2^n})$ it holds that
  \begin{equation}
    \|P_{M}(U \sigma^{(n)} U^{\dagger})- P_{M}(2^{-n} P_{2^n})\|_1 = \|P_{M|_{\Pi_{2^n}}}(U \sigma^{(n)} U^{\dagger})- P_{M|_{\Pi_{2^{n}}}}(2^{-n} P_{2^n})\|_1.
  \end{equation}
  Therefore, it suffices to consider measurements in $\mathcal{M}|_{\Pi_{2^n}}$, which, by assumption of the theorem, are such that $\mathbf{P}_{\mathcal{M}|_{\Pi_{2^n}}}$ can be $\eta_n$-decododed from~$\mathcal{O}_{n}$. Therefore, we can apply \cref{lem:recovery} to the states~$U \sigma^{(n)} U^{\dagger}$ and $2^{-n} P_{2^n}$. Let $\mathcal{S}_M$ be the semi-algebra and let~$N_F^{M}$ be the operator from \cref{lem:recovery} associated to $M \in \mathcal{M}$ and $F \in \mathcal{S}_{M}$, then 
  \begin{align}\nonumber
        d_{\mathcal{M}}(U \sigma^{(n)} U^{\dagger},&2^{-n} P_{2^n}) = \sup_{M \in \mathcal{M}} \sup_{F \in \mathcal{S}_{M}} |\tr(M(F) (2^{-n} P_{2^n}-U \sigma^{(n)} U^{\dagger}))|\\ \nonumber
        &\leq  2\eta_n + \sup_{M \in \mathcal{M}}  \sup_{F \in \mathcal{S}_{M}} |\tr(N_F^{M} (2^{-n} P_{2^n}-U \sigma^{(n)} U^{\dagger}))|\\  \nonumber
        &=  2\eta_n + \sup_{M \in \mathcal{M}} \sup_{F \in \mathcal{S}_{M}} |\tr(P_{2^n}N_F^{M} P_{2^n} (2^{-n} P_{2^n}-U \sigma^{(n)} U^{\dagger}))| \\ \label{eq:estimate} 
        &\leq 2\eta_n + \sup_{M \in \mathcal{M}}  \sup_{F \in \mathcal{S}_{M}} \sum_{E\in \mathcal{O}} p_E^{M}(F) |\tr(P_{2^n} E P_{2^n} (2^{-n} P_{2^n}-U \sigma^{(n)} U^{\dagger}))|, 
  \end{align} 
  where the sum over $E$ goes over finitely many $p_E^{M}(F) > 0$. For every $\bar{E} = \sum_{j} p_j \ketbra{\psi_j}{\psi_j} \in P_{2^n}\mathcal{O}_n P_{2^n}$ we define a function
  \begin{equation}
    D_{\bar{E}}: U \to \sum_j p_j \, D^{\sigma^{(n)}}_{\ket*{\psi_j}}(U)
  \end{equation}
  where $D^{\sigma^{(n)}}_{\ket*{\psi_j}}(U)$ is given by \cref{eq:def_D_psi}.
  Diagonalizing $P_{2^n} E P_{2^n}$ and applying the triangle inequality for the absolute value yields
  \begin{equation}\label{eq:D_estimate}
      d_{\mathcal{M}}(U \sigma^{(n)} U^{\dagger},2^{-n} P_{2^n}) \leq 2\eta_n + \sup_{M \in \mathcal{M}}  \sup_{F \in \mathcal{S}_{M}} \sum_{E\in \mathcal{O}} p_E^{M}(F) D_{P_{2^n} E P_{2^n}}(U).
  \end{equation}

  We now turn to finding upper bounds on~$D_{\bar{E}}(U)$. To do so, the following two observations are useful. The function $D_{\bar{E}}$ is $\smash{2^{-n+\frac{5}{2}}\tr(\bar{E})}$-Lipschitz continuous with respect to the $2$-norm on~$\smash{U(\mathcal{H}_{2^n})}$, as
  \begin{align}
    \begin{split}
      |D_{\bar{E}}(U)-D_{\bar{E}}(U')| &\leq   \sum_j p_j |D^{\sigma^{(n)}}_{\ket*{\psi_j}}(U)- D^{\sigma{(n)}}_{\ket*{\psi_j}}(U')| \\
      &\leq  \sum_j p_j \frac{4 \sqrt{2}}{2^n +1} \|U-U'\|_2\\
      &\leq 2^{-n+\frac{5}{2}}\tr(\bar{E})\|U-U'\|_2
    \end{split}
  \end{align}
  where for the second inequality we used \cref{lem:lip_const_min} and that $\smash{H_{\min}(\sigma^{(n)}) = -\log(\frac{2}{2^n +1})}$. Furthermore, from \cref{lem:average}, it follows that for a unitary $U$ chosen according to the Haar measure $\langle D_{\bar{E}}(U)\rangle \leq 2^{-\frac{3}{2}n+\frac{1}{2}}\tr(\bar{E})$. 
  
  These two observations about $D_{\bar{E}}(U)$ allow us to apply \cite[Theorem 5.16 and Theorem 5.9]{Meckes_2019}, which states that, for any $d$-dimensional Hilbert space $\mathcal{H}$, for any function $f : U(\mathcal{H}) \to \mathbb{R}$ that is $\kappa$-Lipschitz with respect to the $2$-norm, and for a unitary $U$ chosen according to the Haar measure we have
  \begin{equation}
    \mathrm{Pr}\left(f(U) \geq \langle f \rangle + \varepsilon\right) \leq e^{-\frac{\varepsilon^2d}{24 \kappa^2}}.
  \end{equation}
  Applying this theorem yields
  \begin{equation}
    \mathrm{Pr}\left(D_{\bar{E}}(U) \geq \Delta(\bar{E}) \right) \leq e^{-\frac{\delta_n^2 2^n}{24 \left(4 \sqrt{2}\right)^2}}
  \end{equation}
  where $\Delta(\bar{E}) \coloneq 2^{-n}\tr(\bar{E}) \left(2^{-\frac{n}{2} +\frac{1}{2}} + \delta_n\right)$ and $\delta_n^2 \coloneq \frac{24 \left(4 \sqrt{2}\right)^2}{2^n}(1+\ln(|\mathcal{O}_n|))$. Combined with the union bound this yields
  \begin{align}
    \mathrm{Pr}\big(\exists \bar{E}  \in P_{2^n}\mathcal{O}_n P_{2^n}: D_{\bar{E}}(U) \geq  \Delta(\bar{E})\big) \leq |P_{2^n}\mathcal{O}_n P_{2^n} | e^{-\frac{\delta_n^2 2^n}{24 \left(4 \sqrt{2}\right)^2}} \leq |\mathcal{O}_n | e^{-\frac{\delta_n^2 2^n}{24 \left(4 \sqrt{2}\right)^2}}.
  \end{align}
  As, $|\mathcal{O}_n| \, \mathrm{exp}(-\textstyle{\frac{\delta_n^2 2^n}{24 \left(4 \sqrt{2}\right)^2}}) = e^{-1} < 1$, we have that 
  \begin{equation}
    \mathrm{Pr}(\forall \bar{E}  \in P_{2^n}\mathcal{O}_n P_{2^n}: D_{\bar{E}}(U) \leq \Delta(\bar{E})) > 0
  \end{equation}
  and, thus, there exists a unitary such that $\forall \bar{E}  \in P_{2^n}\mathcal{O}_n P_{2^n}: D_{\bar{E}}(U) \leq \Delta(\bar{E})$. We now define $U^{(n)}$ to be a unitary with this property. From this definition and \cref{eq:D_estimate}, we find
  \begin{align}
    \begin{split}
      d_{\mathcal{M}}(U^{(n)} \sigma^{(n)} (U^{(n)})^{\dagger},2^{-n} P_{2^n})  &\leq  2\eta_n + \sup_{M \in \mathcal{M}}  \sup_{F \in \mathcal{S}_{M}} \sum_{E\in \mathcal{O}} p_E^{M}(F) \Delta(P_{2^n}E P_{2^n}) \\
      &\leq  2\eta_n + \left(\sqrt{2}\times 2^{-\frac{n}{2}} + \delta_n\right) \sup_{M \in \mathcal{M}}  \sup_{F \in \mathcal{S}_{M}} \tr( 2^{-n} P_{2^n} N_F^{M}) \\
      &\leq  3\eta_n + \left(\sqrt{2}\times 2^{-\frac{n}{2}} + \delta_n\right) 
    \end{split}
  \end{align}
  By assumption of the theorem, we know that $\lim_{n \to \infty} \frac{\ln(|\mathcal{O}_n|)}{2^n} = \ln(2)\lim_{n \to \infty} \varepsilon_n = 0$, which implies that $\lim_{n \to \infty} \delta_n = 0$.
  Therefore, $\lim_{n \to \infty} d_{\mathcal{M}}(U^{(n)}\sigma^{(n)} (U^{(n)})^{\dagger}, 2^{-n}P_{2^n}) = 0$, which was what remained to prove the theorem.
\end{proof}

\subsection{Proof of \texorpdfstring{\cref{rem:upper_bound}}{}}
We recall the definition of $\varepsilon$-nets and one of their properties.
\begin{definition}[$\varepsilon$-net of states]
  Let $\mathcal{H}$ be a Hilbert space. Then we call a set of states $\mathcal{N} \subset \mathcal{H}$ an \emph{$\varepsilon$-net} if 
  \begin{equation}\label{eq:net}
  \forall \ket{\varphi} \in \mathcal{H} \ \exists \ket*{\bar{\varphi}} \in N\!:  \|\ket{\varphi} - \ket*{\bar{\varphi}}\| \leq \varepsilon.
  \end{equation}
\end{definition}

\begin{lemma}[Size of $\varepsilon$-Net]\label{lem:epsilon_net}
  Let $\mathcal{H}$ be a $d$-dimensional Hilbert space, then there exists a $\varepsilon$-net $\mathcal{N}$ on $\mathcal{H}$ such that 
  \begin{equation}
    |\mathcal{N}| \leq  \left(1+\frac{2}{\varepsilon}\right)^{2d}.
  \end{equation}
\end{lemma}

\begin{proof}
  \cite[Theorem 1.8]{Watrous2018}

\end{proof}

\begin{lemma}\label{lem:norm_conversion}
  Let $\mathcal{H}$ be a Hilbert space, then for any two pure states $\ket{\phi}, \ket{\psi} \in \mathcal{H}$ it holds that 
  \begin{equation}
     \|\ketbra{\phi} - \ketbra{\psi}\|_1 \leq 2\|\ket{\phi} - \ket{\psi}\|.
  \end{equation}
\end{lemma}
\begin{proof}
  First note that
  \begin{align}
    \begin{split}
      \|\ket{\psi}- \ket*{\phi}\|^2 &= 2-2\mathrm{Re}(\braket{\psi}{\phi}) \\
      &= 2-2\mathrm{Re}(\braket{\psi}{\phi}) -1 +|\braket{\psi}{\phi}|^2 + (1-|\braket{\psi}{\phi}|^2) \\
      &\geq 2-2\sqrt{\mathrm{Re}(\braket{\psi}{\phi})^2 + \mathrm{Im}(\braket{\psi}{\phi})^2 }  -1 +|\braket{\psi}{\phi}|^2 + (1-|\braket{\psi}{\phi}|^2) \\
      &\geq 1-2|\braket{\psi}{\phi}| + |\braket{\psi}{\phi}|^2  + (1-|\braket{\psi}{\phi}|^2) \\
      &=(1-|\braket{\psi}{\phi}|)^2  + (1-|\braket{\psi}{\phi}|^2) \\
      &\geq 1-|\braket{\psi}{\phi}|^2 \\
      &= \frac{1}{2} \|\ketbra{\psi}- \ketbra*{\phi} \|^2_2.
    \end{split}
  \end{align}

  Note that $\ketbra{\psi}-\ketbra*{\phi}$ is a Hermitian rank-$2$ operator, thus there exist states $\ket{e_1},\ket{e_2}$ such that 
  \begin{equation}
   \ketbra{\psi}-\ketbra*{\phi} = a_1 \ketbra{e_1} + a_2\ketbra{e_2}
  \end{equation}
  and, we find that 
  \begin{align}
    \begin{split}
      \| \ketbra{\psi}-\ketbra*{\phi}\|_1 &= |a_1| + |a_2| \\
      &= \langle (1,1),(|a_1|, |a_2|)\rangle \\
      &\leq \sqrt{2}(|a_1|^2 + |a_2|^2)^{1/2} \\
      &= \sqrt{2} \| \ketbra{\psi}-\ketbra*{\phi}\|_2.
    \end{split}
  \end{align}
  Combining this result with the above, we find
   \begin{equation}
     \|\ketbra{\phi} - \ketbra{\psi}\|_1 \leq 2\|\ket{\phi} - \ket{\psi}\|.
  \end{equation}
\end{proof}

\begin{lemma}\label{lem:entropy_all}
  Let $\mathcal{M}_{\mathrm{all}}$ be the set of all measurements. Then $\mathbf{P}_{\mathcal{M}_{\mathrm{all}}}$ has asymptotic ent\-ropy at most $(2^{2^{n+\log(n)+3}})_{n \in \mathbb{N}}$. 
\end{lemma}

\begin{proof}
  Let $\{\ket{n}\}_{n \in \mathcal{N}}$ be a basis,~$\mathcal{H}_{2^n}$ the span of the first $2^n$ basis elements, $\Pi_{2^n} = \sum_{n = 0}^{2^n-1} \ketbra{n}$, $M$ a measurement, $M|_{\Pi_{2^n}}$ the restriction of $M$ to the support of $\Pi_{2^n}$, and~$\mathcal{N}_n$ a $(2^{-2n})$-net on~$\mathcal{H}_{2^n}$. 
  
  For any measurable set $F \in \Sigma_{M}$, consider $M|_{\Pi_{2^n}}(F)$, which has support only on~$\mathcal{H}_{2^n}$. We decompose~$M|_{\Pi_{2^n}}(F)$ into its eigenbasis $M|_{\Pi_{2^n}}(F) = \sum_i p_i \ketbra{\psi_i}$. Let $\ket{\phi_i} \in \mathcal{N}_n$ be such that $\|\ket{\phi_i} - \ket{\psi_i}\| \leq 2^{-2n}$ and define $E_F^{M} = \sum_i p_i \ketbra{\phi_i}$. The corresponding decoding operation~$\mathcal{D}_{M}$ is defined in the obvious way.
  We find
  \begin{align}
    \begin{split}
      |\tr((M|_{\Pi_{2^n}}(F)-E_F^{M})\rho)| &\leq \|M|_{\Pi_{2^n}}(F)-E_F^{M}\|_1 \\
      &\leq \sum_i p_i \| \ketbra{\psi_i} - \ketbra{\phi_i}\|_1 \\
      &\leq 2^{-2n+1}  \tr(M|_{\Pi_{2^n}}(F)) \leq 2^{-n +1}
    \end{split}
  \end{align}
  where we used \cref{lem:norm_conversion} in the second to last inequality. 
  
  The same argument can be applied for every measurement $M$ and every $F \in \Sigma_M$. Thus, using for every measurement $M$ the sequence of decoding operations such that $\mathcal{D}^{(k)}_{M} = \mathcal{D}_{M}$, the representation~$\mathbf{P}_{\mathcal{M}_{\mathrm{all}}|_{\Pi_{2^n}}}$ can be $2^{-n+2}$-decoded from~$\mathcal{N}_n$, where one understands the elements of $\mathcal{N}_n$ as the corresponding rank-$1$ projectors. From \cref{lem:epsilon_net} it follows that there is a $2^{-2n}$-net~$\mathcal{N}_n$ such that 
  \begin{equation}
    |\mathcal{N}_{n}| = (1+2^{2n+1})^{2^{n+1}} \leq 2^{(2n+2)2^{n+1}}  = 2^{2^{n+\log(n)+3}}.
  \end{equation}
  Therefore, the asymptotic entropy of $\mathbf{P}_{\mathcal{M}_{\mathrm{all}}}$ is at most $(2^{2^{n+\log(n)+3}})_{n \in \mathbb{N}}$.
\end{proof}

\begin{restatable}{corollary}{maxentropy}\label{rem:upper_bound}
    The asymptotic entropy of any representation $\mathbf{P}_\mathcal{M}$ is at most $(2^{n+\log(n)+3})_{n \in \mathbb{N}}$.
\end{restatable}

\subsection{Proof of \texorpdfstring{\cref{cor:prod_conv}}{}}

\prodConv*

\begin{proof}
  Denote the two subsystems relative to which $\mathcal{M}_{\otimes}$ is product to by $A$ and $B$. Furthermore, let $\{\ket{n}_A\}_{n \in \mathbb{N}}, \{\ket{n}_B\}_{n \in \mathbb{N}}$ be orthonormal bases of $\mathcal{H}_A$ and $\mathcal{H}_B$ respectively. We denote by $\mathcal{H}_{N,A}$ the span of the first $N$ basis vectors and by $\Pi_{N,A}$ the projector onto~$\mathcal{H}_{N,A}$, and analogously for $\mathcal{H}_{N,B}$ and $\Pi_{N,B}$. On the joint system $AB$, we define the family of projectors~$\{\Pi_{2^n}\}_{n \in \mathbb{N}}$ as $\Pi_{2^n} = \Pi_{(2^{n/2}),A} \otimes \Pi_{(2^{n/2}),B}$ if $n$ is even and $\Pi_{2^n} = \Pi_{2^{(n+1)/2},A} \otimes \Pi_{(2^{(n-1)/2}),B}$ if~$n$ is odd. Let $\smash{\mathcal{N}^{(n)}_A}$ and $\smash{\mathcal{N}^{(n)}_B}$ be $2^{-2n}$-nets on the supports of the corresonding subspaces. Using these nets, we define the set of positive operators
  \begin{equation}
    \mathcal{O}_n \coloneq \bigl\{\ketbra{\psi}_A\otimes \ketbra{\phi}_B|\ket{\psi}_A \in \mathcal{N}^{(n)}_A, \ket{\phi}_B \in \mathcal{N}^{(n)}_B \bigr\}.
  \end{equation}
  
  For any measurement $M = M_A \otimes M_B \in \mathcal{M}_{\otimes}$, the associated $\sigma$-algebra, $\Sigma_{M_A \otimes M_B}$, is the product $\sigma$-algebra of the $\sigma$-algebras associated to the measurements $M_A$ and $M_B$, i.e., $\Sigma_{M_A \otimes M_B} = \Sigma_{M_A} \otimes \Sigma_{M_B}$. A product $\sigma$-algebra admits a semi-algebra of rectangles $\mathcal{S}_{AB} \coloneq \{ A\times B| A \in \Sigma_{M_A}, B \in \Sigma_{M_B}\}$. Denote by $\mathcal{S}_{AB}^{*}$ the semi-algebra which corresponds to the closure of $\mathcal{S}_{AB}$ under finite disjoint unions. Any $F \in \mathcal{S}_{AB}^{*}$ is the finite disjoint union of rectangles, i.e., $F = \bigsqcup_{i = 0}^{k} A_i \times B_i$. Therefore, the effect $M|_{\Pi_{2^n}}(F)$ can be written as 
  \begin{equation}
    M|_{\Pi_{2^n}}(F) = \sum_{k = 1}^{\ell}E_{A,k} \otimes E_{B,k},
  \end{equation}
  with $E_{A,k}$ and $E_{B,k}$ effects, such that the support of $E_{A,k} \otimes E_{B,k}$ is contained in the support of $\Pi_{2^n}$. Based on their diagonalization $E_{A,k} \otimes E_{B,k} = \sum_i p_{i,k} \ketbra{\psi_i}_{A,k} \otimes \sum_j q_{j,k} \ketbra{\phi_i}_{B,k}$, with $p_{i,k} ,q_{i,k} \geq 0$, we define
  \begin{equation}
    N_{F}^{M} = \sum_{k = 1}^{\ell} \sum_i p_{i,k} \ketbra*{\bar{\psi_i}}_{A,k} \otimes \sum_j q_{j,k} \ketbra*{\bar{\phi_i}}_{B,k}
  \end{equation}
  where $\ketbra*{\bar{\psi_i}}_{A,k} \in \mathcal{N}^{(n)}_A$ and $\ketbra*{\bar{\phi_i}}_{B,k} \in \mathcal{N}^{(n)}_B$ such that $\|\ket*{\bar{\psi_i}}_{A,k} - \ket{\psi_i}\| \leq 2^{-2n}$ and analogously for $\ket*{\bar{\phi}_i}_{B,k}$. The collection of operators $\{N^{M}_F\}_{F \in \mathcal{S}_{AB}^{*}}$ defines a decoding operation $\mathcal{D}_M((P_{E})_{E \in \mathcal{O}_n})$ in the obvious way. 
  Then we find that for all $F \in \mathcal{S}_{AB}^{*}$
  \begin{align}
    \begin{split}
      |\tr((M|_{\Pi_{2^n}}(F)- N_{F}^{M})\rho)| &\leq \|M|_{\Pi_{2^n}}(F)- N_F^{M}\|_1 \\
      &\leq \sum_{i,j,k} p_{i,k} q_{j,k} \|\ketbra{\psi_i}_{A,k}
      \otimes \ketbra{\phi_i}_{B,k} - \ketbra*{\bar{\psi_i}}_{A,k} \otimes \ketbra*{\bar{\phi_i}}_{A,k}\|_1 \\ 
      &\leq \sum_{i,j,k} p_{i,k} q_{j,k} \big(\|\ketbra{\psi_i}_{A,k}
      \otimes \big(\ketbra{\phi_j}_{B,k} - \ketbra*{\bar{\phi_j}}_{B,k}) \|_1  \\[-2ex]
      &\qquad \qquad \quad \,  + \|(\ketbra{\psi_i}_{A,k} - \ketbra*{\bar{\psi_i}}_{A,k}) \otimes \ketbra*{\bar{\phi_j}}_{B,k}\|_1 \big)\\
      &\leq 2^{-2n+2} \tr(M|_{\Pi_{2^n}}(F)) \leq 2^{-n+2}.
    \end{split}
  \end{align}
  As $M \in \mathcal{M}_{\otimes}$ was arbitrary, and for each $M$, choosing the corresponding sequence of decoding operations to consist only of $\mathcal{D}_M((P_{E})_{E \in \mathcal{O}_n})$, it follows that $\mathbf{P}_{\mathcal{M}_{\otimes}|\Pi_{2^n}}$ can be $2^{-n+3}$-decoded from $\mathcal{O}_n$.

  By \cref{lem:epsilon_net} there exist nets $\mathcal{N}^{(n)}_A, \mathcal{N}^{(n)}_B$ such that
  \begin{equation}
    |\mathcal{O}_{n}| \leq 2^{(2n +2)2^{\frac{n+1}{2}+2}} = 2^{\left((2n +2) 2^{-\frac{n-1}{2}+2}\right)2^{n}}.
  \end{equation}
  This implies that the asymptotic entropy of $\mathbf{P}_{\mathcal{M}_{\otimes}}$ is at most $(2^{((2n +2) 2^{-\frac{n-1}{2}+2})2^{n}})_{n \in \mathbb{N}}$. Thus, by \cref{thm:entropy}, the representation $\mathbf{P}_{\mathcal{M}_{\otimes}}$ is not robust.
\end{proof}

\begin{remark}
  Note that an analogous proof also shows that the set $\mathcal{M}_{\mathrm{sep}}$ of measurements whose effects are separable does not lead to a robust representation. In particular, this conclusion holds for any subset of $\mathcal{M}_{\mathrm{sep}}$, e.g, the set of measurements that can be implemented with local operations and classical communication (LOCC).
\end{remark}

\subsection{Related results}
\Cref{thm:entropy} characterized the robustness of a representation $\mathbf{P}_\mathcal{M}$ by how compressible it is. Using \cite[Theorem 1]{Aubrun2015} one can, additionally, find a characterization that is based on the asymptotic number of measurements in $\mathcal{M}$.

\begin{definition}\label{def:asymptotic_size}
    Let $\{\Pi_{2^n}\}_{n \in \mathbb{N}}$ be a nested family of projectors with $\mathrm{rank}(\Pi_{2^n}) \geq 2^n$, then the \emph{asymptotic size of $\mathcal{M}$ is at most} $(|\mathcal{M}|_{\Pi_{2^n}}|)_{n \in \mathbb{N}}$. 
\end{definition}

\begin{theorem}\label{thm:size}
    If there exists a zero-sequence $(\varepsilon_n)_{n \in \mathbb{N}}$ such that $\mathcal{M}$ has asymptotic size at most~$(2^{\varepsilon_n 2^{2 n}})_{n \in \mathbb{N}}$, then $\mathbf{P}_\mathcal{M}$ is not robust.
\end{theorem}

\begin{proof}
  The assumption $|\mathcal{M}_{\Pi_{2^n}}| \leq 2^{\varepsilon_n 2^{2 n}}$ implies that $|\mathcal{M}_{\Pi_{2^n}}| \leq e^{\ln(2)\varepsilon_n 2^{2n}}$. From \cite[Theorem 1]{Aubrun2015} it follows that there are constants $C,c > 0$ such that there exists a Hermitian trace-class operator $A^{(n)}$ with $\|A^{(n)}\|_{\mathcal{M}} \leq \Delta_n \|A^{(n)}\|_{1}$ where $\smash{\Delta_n = \max(\sqrt{\frac{\ln(2) \varepsilon_n}{c}},C 2^{- \frac{n}{2}})}$. Therefore, the norms $\|\cdot \|_{\mathcal{M}}$ and $\| \cdot \|_1$ are not equivalent. From \cref{prop:topology}, it follows that~$\mathbf{P}_{\mathcal{M}}$ is not robust. 
\end{proof}

Combined with \cref{rem:upper_bound}, this theorem shows that a robust representation $\mathbf{P}_{\mathcal{M}}$ is highly compressible: the number of measurements in $\mathcal{M}_{\Pi_{2^n}}$, and thus also the number of entries in $\mathbf{P}_{\mathcal{M}}$, needs to scale at least as $2^{2^{2n}}$, but there exists an encoding into a set $\mathcal{O}$ with only of the order of $2^{n 2^n}$ elements. 

\section{Proof of \texorpdfstring{\cref{thm:GPT}}{}}\label{app:GPT}
Before we go on to prove \cref{thm:GPT}, we review the necessary aspects of GPTs for this paper.
\subsection{GPTs}
\paragraph{States and measurements.}
Operationally, a state of a GPT system is understood as a particular preparation procedure. Consequently, the state space is the set of all preparation procedures. It is assumed that the state space is convex. The convex mixture $p\rho + (1-p) \sigma$ of two states $\rho, \sigma$ is operationally understood as the procedure that prepares $\rho$ with probability $p$ and $\sigma$ with probability~$(1-p)$.\footnote{An important assumption here is that the randomness that determines whether $\sigma$ or $\rho$ is prepared is independent of any randomness that is used in the preparation procedure of $\rho$ or $\sigma$.} Furthermore, two preparation procedures that lead to the same probability distributions for all measurements that can be performed on this system, are treated as the same state. This motivates the following definitions. 

\begin{definition}\label{def:GPT_state_space}
    The \emph{state space} of a system $A$ is a convex subset $\mathcal{S}_A \subset V_A$ of an $\mathbb{R}$-vector space $V_A$\footnote{The dimension of this vector space may be unbounded.} such that there exists a linear function $\mathbf{1}_A: V_A \to \mathbb{R}$ with the property $\mathbf{1}_A(\mathcal{S}_A) = 1$.

    The set of \emph{effects} $\mathcal{E}_A$ of system $A$ are a subset of the dual space $\mathcal{E}_A \subseteq V_A^{*}$ such that $\forall E \in \mathcal{E}_A: E(\mathcal{S}_A) \subseteq [0,1]$ and 
    \begin{equation}
      \forall \omega_1, \omega_2 \in \mathcal{S}_A: \left(\forall E \in  \mathcal{E}_A: \ E(\omega_1) = E(\omega_2)\right) \implies \omega_1 = \omega_2
    \end{equation}
    as well as 
    \begin{equation}
      \forall E_1, E_2 \in \mathcal{E}_A: \left(\forall \omega \in  \mathcal{S}_A: \ E_1(\omega) = E_2(\omega)\right) \implies E_1 = E_2.
    \end{equation}
    We call a set of effects $\{E_{i}\}_{i = 1}^{N} \subseteq \mathcal{E}_A$ a \emph{measurement} if 
    \begin{equation}
        \sum_{i = 1}^{N} E_{i} = \mathbf{1}_A.
    \end{equation}
\end{definition}

As in quantum theory, we define a tomographically complete set of measurements.
\begin{definition}
  A set of measurements $\mathcal{M}$ is \emph{tomographically complete} if the map 
  \begin{equation}
    v \in \mathcal{S}_A \mapsto (P_{M}(v))_{M \in \mathcal{M}}
  \end{equation}
  is injective, where $P_M(\rho)$ is the probability distribution of the measurement $\mathcal{M}$.
\end{definition}

\paragraph{System composition.}
There is no single rule how the state spaces of two systems $A$ and $B$ compose that applies to any GPT. The only requirement is that states can be independently composed. Technically, this means that for any two systems $A,B$ there exists a bilinear map
\begin{align}
  \begin{split}
    \imath: V_A \times V_B &\to V_{AB}, \\
    (\omega_A,\omega_B) &\mapsto \omega_A\omega_B,
  \end{split}
\end{align}
such that $\mathrm{Im}(\imath|_{\mathcal{S}_A \times \mathcal{S}_B}) \subset \mathcal{S}_{AB}$ and for any two effects $E_{A}\in\mathcal{E}_A, E_{B}\in\mathcal{E}_B$ there is an effect $E_{A} E_{B}\in\mathcal{E}_{AB}$ such that
\begin{equation}
    (E_{A} E_{B})(\omega_A\omega_B) = E_{A}(\omega_A)\,E_{B}(\omega_B).
\end{equation}
Often one additionally requires that the composition rule satisfies the so-called \emph{local tomography} assumption.
\begin{definition}\namedlabel{ass:local_tomography}{local tomography} 
  A GPT satisfies the \emph{local tomography} assumption if, for any two systems $A$ and $B$, the set of local measurements $\mathcal{M}_{\otimes}$ is tomographically complete.
\end{definition}

\paragraph{Metric.} In most treatments of GPTs there is no explicit metric defined on the state space. Motivated by the trace distance, we use a metric defined in analogy to the trace distance. A similar metric was already introduced in \cite{Chiribella2010,Chiribella2011,DAriano_2017} for operational probabilistic theories, a close relative to GPTs. 
\begin{definition}
  The \emph{trace distance} $\delta$ on the state space $\mathcal{S}_A$ is defined by
  \begin{equation}
    \delta(\rho,\sigma) \coloneq \frac{1}{2}\sup_{M \in \mathcal{M}_A} \|P_M(\rho)-P_M(\sigma)\|_1.
  \end{equation}
\end{definition}
\noindent

 This metric $\delta$ satisfies the composability criterion~\eqref{eq:stability} if the system $A$ has the property that for all systems $B$
\begin{equation} \label{eq:compprod}
  \forall \rho_B \in \mathcal{S}_B, M \in \mathcal{M}_{AB}: \ M(\cdot \otimes \rho_B) \in \mathcal{M}_{A}.
\end{equation}
As we did in quantum theory, we can also define a metric associated to a tomographically complete measurement set $\mathcal{M}$.
\begin{definition}
  For any tomographically complete set of measurements $\mathcal{M}$ we define the metric
  \begin{equation}
    d_{\mathcal{M}}(\rho,\sigma) \coloneq \frac{1}{2} \sup_{M \in \mathcal{M}} \|P_M(\rho)-P_M(\sigma)\|_1.
  \end{equation}
\end{definition}
\noindent
As in quantum theory, we define the notion of a fragile measurement set~$\mathcal{M}$.
\begin{definition}
  A measurement set $\mathcal{M}$ on a system $A$, is called \emph{fragile} there is an effects $M \in \mathcal{E}_A$ such that the topologies $d_{\mathcal{M}}$ and $d_{\mathcal{M} \cup \{(E, \mathbf{1}_A - E)\}}$ are different. 
\end{definition}

\subsection{Constructing the GPT for \texorpdfstring{\cref{thm:GPT}}{}}
We now construct the GPT that proves \cref{thm:GPT}. This GPT has two elementary types of systems: keys and locks. All other types of systems are obtained by composing key and lock systems. We start by introducing the elementary systems.

\paragraph{Lock systems.} Before we give the formal definition, we give the intuition behind a lock system: A lock system acts like a lock that takes a bit string as input and opens if this bit string matches an internally stored one. More technically, we model a lock as a system where for every\footnote{We denote by $\{0,1\}^*$ the set of all bit strings, including the empty bit string.} bit string $s \in \{0,1\}^*$, corresponding to the input to the lock, there is a measurement consisting of two effects $\EL{s}{\yes}$ and $\EL{s}{\no}$, corresponding to the outcome that the lock opens or stays closed, respectively. Furthermore, for every bit string~$k$ there is a state of the lock $\SL{k}$ where the lock opens upon the input~$k$. Graphically, we depict a measurement with input $s$ on this state by 
\begin{equation}
  \EL{s}{\yes}(\SL{k}) = \lock{k}{s}{\yes} = 1.
\end{equation}
To fully specify the state $\SL{k}$, we also need to define the behaviour of the lock when the input $s$ is different from $k$. If the input bit string $s$ is longer than $k$, the lock opens if the first $|k|$ bits of $s$ match with those of $k$ --- the lock simply ignores the superfluous input. If $s$ is shorter than $k$, then the lock just randomly generates more bits, appends them to $s$, and checks if this new bit string agrees with $k$. In summary, 
\begin{equation}
 \EL{s}{\yes}(\SL{k}) = \lock{k}{s}{\yes} = \begin{cases}
    \delta_{s,k} & \text{if } |s| \geq |k|  ,\\[4pt]
    \frac{\delta_{s,k}}{2^{|k|-|s|}}
     & \text{else }
  \end{cases}
\end{equation}
where $|\cdot|$ denotes the length of a bit string and $\delta_{s,k}=1$ if deleting the trailing bits of the longer bit string results in two identical bit strings, and $\delta_{s,k}=0$ otherwise. Sometimes locks can be stubborn, and they do not open no matter what you do.\footnote{Maybe an angle grinder would open it, but we do not want to harm our object of study.} We model this behaviour by a state $\SL{\bot}$ that does not open for any input, i.e., 
\begin{equation}
  \forall s \in \{0,1\}^*: \, \EL{s}{\yes}(\SL{\bot}) = \lock{\bot}{s}{\yes} = 0.
\end{equation} 
The state space of a lock system is then the convex hull of $\{\SL{k}\}_{k \in \{0,1\}^* \cup \{\bot\}}$. 

Before we state the formal definition, we must introduce some notation. We denote by~$\chi_{I}$ is the characteristic function of the set $I$. For a bit string $s \in \{0,1\}^n$, we define the interval $I_s = [0.s,0.s+2^{-n}]$ where $0.s$ understood as the rational number $r$ with $0.s$ as its binary expansion.
\begin{definition} The state space is of a \emph{lock system} is 
    \begin{equation}
      \mathcal{S}_L = \left\{(f,1)| f\in \mathrm{conv}(\{\chi_{I_s}| s \in \{0,1\}^*\} \cup \{0\}) \right\} \subset V_L
    \end{equation}
    with  $V_L = L^{1}([0,1]) \oplus \mathbb{R}$ and the $\mathbf{1}_L$-effect is given by $\mathbf{1}_L(f,c) = c$.
   
    For every $s \in \{0,1\}^*$ and $r \in \{\yes, \no\}$ there is an effect $\EL{s}{r}$. The set of effects $\mathcal{E}_L$ is given by all convex combination of these effects. The action of the effect $\EL{s}{r}$ on an element of $V_L$ is
    \begin{align}
      \EL{s}{\yes}(f,c) &= \frac{1}{|I_s|}\int_{I_s} f(r) \, dr \\
      \EL{s}{\no} &= \mathbf{1}_L - \EL{s}{\yes}
    \end{align}
\end{definition}
It is easy to see that this combination of states and effects satisfies \cref{def:GPT_state_space}. The states $\SL{k}$ and $\SL{\bot}$ we intuitively introduced before, are formally given by
\begin{align}
  \SL{k} &\coloneq  (\chi_{I_k}, 1) \\
  \SL{\bot} &\coloneq (0, 1).
\end{align}
Indeed, these states have the desired behaviour when measured, as
\begin{align}
  \EL{s}{\yes}(\SL{\bot}) &= 0 \\
  \EL{s}{\yes}(\SK{k}) &= \frac{1}{|I_s|} \int_{I_s} \chi_{k} \, dx = \frac{|I_s \cap I_k|}{|I_s|} = \begin{cases}
    \delta_{s,k} \text{ if } |s| \geq |k|  \\
    \frac{\delta_{s,k}}{2^{|k|-|s|}} \text{ else}
  \end{cases}.
\end{align}

\paragraph{Key systems.}
Intuitively, a key system is a system that has an internally stored bit string, the key. The system can be queried to output the first $n$ bits of the key, where $n$ is any natural number (including $0$). More technically speaking, there is a measurement consisting of $2^n$ effects $\{\EK{n}{s}\}_{s \in \{0,1\}^n}$, corresponding to the $2^n$ possibilities for the first $n$ bits of the key. For every bit string $k \in \{0,1\}^{*}$, there is a state of the key $\SK{k}$, such that if the first $n$ bits of the key are measured, the output is the first $n$ bits of $k$. In particular, if $n = |k|$, we have
\begin{equation}
  \EK{|k|}{k}(\SK{k}) = \key{k}{|k|}{k} = 1.
\end{equation}
If $n > |k|$, the key system randomly generates bits and appends them to $k$ until the resulting bit string has length $n$. In particular, if a key system has no key stored, i.e., it is in state~$\SK{\emptyset}$, measuring the first $n$ bits of the key yields a uniform distribution over all $n$-bit strings. We summarize the behaviour of these states
\begin{align}
  \EK{n}{s}(\SK{k}) = \key{k}{n}{s} = \begin{cases}
    \delta_{s,k} &\text{ if } |k| \geq n \\
    \frac{\delta_{s,k}}{2^{n-|k|}} &\text{ else }
  \end{cases}.
\end{align}
The state space of a key system is then defined as the convex hull of $\{\SK{k}\}_{k \in \{0,1\}^*}$. To make this a valid state space, we need to ensure that states which cannot be distinguished by two measurements are identical. The following formal definition takes care of this.
\begin{definition}[Key systems $K$]
   The state space $\mathcal{S}_K$ is 
    \begin{equation}
      \mathcal{S}_K = \mathrm{conv}\left(\left\{\frac{\chi_{I_s}}{|I_s|} \Big| s \in \{0,1\}^*\right\}\right) \subset L^{1}([0,1])
    \end{equation}
    The unit effect $\mathbf{1}_K$ is $\mathbf{1}_K(f) = \int_{[0,1]} f(x) \, dx$. The set of effects $\mathcal{E}_K$ is the convex hull of effects of the form~$\sum_{s \in R \subset \{0,1\}^n} \EK{n}{s}$. The action of an effect~$\EK{n}{s}$ on $f \in L^{1}([0,1])$ is
    \begin{equation}
      \EK{n}{s}(f) = \int_{I_s} f(x) \, dx.
    \end{equation}
\end{definition}

It is easy to see that, \cref{def:GPT_state_space} is satisfied by this state space and set of effects. The states $\SK{k}$ and $\SK{\emptyset}$ we have intuitively introduced are formally given by
\begin{align}
  \SK{k} \coloneq \frac{1}{|I_k|} \chi_{I_k}
\end{align}
where for $k = \emptyset, I_k = [0,1]$. These states indeed have the desired behaviour when measured
\begin{align}
  \EK{n}{s}(\SK{k}) &= \int_{I_s} \chi_{I_k}(x) \, dx =\frac{|I_s \cap I_k|}{|I_k|} = \begin{cases}
    \delta_{s,k} &\text{ if } |k| \geq n \\
    \frac{\delta_{s,k}}{2^{n-|k|}} &\text{ else}.
  \end{cases}
\end{align}

 \paragraph{Composition of keys and locks.} To define the composition of key and lock systems, we consider the smallest possible composed state space that still satisfies the requirements of the GPT framework. This means that we only allow for mixtures of product states. On the measurement side, we allow only local measurements and measurements that can be implemented by first measuring one system and then determine the measurement on the next system based on this outcome. 
 \begin{definition}[Composition rule]\label{def:composition} 
  When composing a key and a lock system the function~$\imath$ is the tensor product $\otimes: V_K \times V_L \to V_K \otimes V_L$. The state space of the composed system is defined as the convex hull of $\mathcal{S}_{KL} \coloneq \mathrm{conv}(\mathcal{S}_{K} \otimes \mathcal{S}_L)$. The identity effect is $\mathbf{1}_{KL} \coloneq \mathbf{1}_{K} \otimes \mathbf{1}_{L}$. The set of effects $\mathcal{E}_{KL}$ are all linear functionals $E$ such that $E(\mathcal{S}_{KL}) \subseteq [0,1]$ and 
  \begin{align}
    \begin{split}
      \exists n \in \mathbb{N},\;  \{p_i\}_{i\in \{1,\dots,n\}} \in \mathbb{R}_{+},\; \{E_{K,i}\}_{i\in \{1,\dots,n\}} \subset \mathcal{E}_K,\; &\{E_{L,i}\}_{i\in \{1,\dots,n\}} \subset \mathcal{E}_L: \\
      E &= \sum_{i= 1}^{n}p_i E_{K,i} \otimes E_{L,i} \\  \text{ or } \quad E &= \mathbf{1}_{KL} - \sum_{i=1}^np_i E_{K,i} \otimes E_{L,i}
    \end{split}
  \end{align}
  The composition of multiple keys and of multiple lock systems is defined analogously.
\end{definition}

One important measurement of a key-lock system is the measurement that uses the output of a measurement on the $K$ system to determine the input on the $L$ system. It essentially measures whether the input that opens the lock is the key that is stored in the key system. Mathematically, this measurement is given by the effects 
\begin{align}
  \begin{split} \label{eq:jointeffect}
\EKL{n}{\yes} & \coloneq \sum_{s\in\{0,1\}^n} \EK{n}{s}\otimes \EL{s}{\yes} \\
\EKL{n}{\no} & \coloneq \mathbf{1}_{KL} - \EKL{n}{\yes}.
  \end{split}
\end{align}
 On a state $\SK{k}\otimes \SL{k'}$ with $k,k' \in \{0,1\}^n$ this measurement acts as
\begin{equation}
  \EKL{n}{\yes}(\SK{k}\otimes \SL{k'})  = \wireKtoL{k}{n}{s}{k'}{\yes} = \delta_{k,k'}.
\end{equation}

\begin{remark}
  The metric $\delta$ satisfies the composability property of \cref{eq:stability}. This follows from~\eqref{eq:compprod} because any measurement is a sum of product effects with positive coefficients. 
\end{remark}

The following lemma shows that this GPT satisfies \ref{ass:local_tomography}.
\begin{lemma}\label{lem:local_tomography}
  Consider a GPT where for every bipartite system $AB$ the joint state space is given by linear combinations of product states, i.e., 
  \begin{equation}
    S_{AB} \subset \mathrm{span}(\imath(\mathcal{S}_A, \mathcal{S}_B)),
  \end{equation}
  then this GPT satisfies local tomography.
\end{lemma}
\begin{proof}
  For every system $A$ the linear map
  \begin{equation}
    M_A: v \in \mathrm{span}(\mathcal{S}_A) \mapsto (E(v))_{E \in \mathcal{E}_A}
  \end{equation}
  is invertible. We define the map $M_AM_B$ on the joint state space by
  \begin{equation}
    M_AM_B\imath(\sigma_A,\sigma_B) = (E_A(\sigma_A)E_B(\sigma_B))_{E_A \in \mathcal{E}_A, E_B \in \mathcal{E}_B}
  \end{equation}
  and demanding linearity.
  Note that the vector space spanned by $(E_A(\sigma_A)E_B(\sigma_B))_{E_A \in \mathcal{E}_A, E_B \in \mathcal{E}_B}$ is isomorphic to the vector space spanned by $(E_A(\sigma_A))_{E_A \in \mathcal{E}_A} \otimes (E_B(\sigma_B))_{E_B \in \mathcal{E}_B}$. Let us denote the isomorphism by $\Phi$. We then find that 
  \begin{equation}
    M_A^{-1} \otimes M_B^{-1} \circ \Phi \circ M_AM_B\imath(\sigma_A,\sigma_B) = \sigma_A \otimes \sigma_B.
  \end{equation}
  By definition of the tensor product there exists a linear map $h$ such that $h(\sigma_A \otimes \sigma_B) = \imath(\sigma_A,\sigma_B)$. Thus, as all maps are linear, $h \circ M_A^{-1} \otimes M_B^{-1} \circ \Phi$ is the inverse of  $M_AM_B$. Hence, a state of the system $AB$ is uniquely determined by the statistics of local measurements. As the systems $AB$ were generic, the theory is locally tomographic.
\end{proof}

\paragraph{The topology of keys and locks.}
We are now ready to show that for this GPT the topology induced by $d_{\mathcal{M}_\otimes}$ on the state space is different form the topology induced by $\delta$. The intuition behind this proof to consider the sequence of states $(\rho_{KL}^{(n)})_{n \in \mathbb{N}}$ given by 
\begin{equation}
  \rho_{KL}^{(n)} \coloneq \sum_{k \in \{0,1\}^n} \key{k}{}{} \otimes \lock{k}{}{}\,. 
\end{equation}
Using only local measurements this state is close to the state $\SK{\emptyset} \otimes \SL{\bot}$. Intuitively, to distinguish these two states, one tries to provoke an opening of the lock. Using only local measurements, one needs to guess the input that opens the lock. This guess is correct only with probability $2^{-n}$. Therefore, the distance between $\rho_{KL}^{(n)}$ and ${\SK{\emptyset} \otimes \SL{\bot}}$ measured by $d_{\mathcal{M}_{\otimes}}$ vanishes as $n \to \infty$. However, if any measurement can be used, then the measurement that reads out the first $n$ bits of the key and uses this as an input to the lock opens the lock with certainty. Therefore, for any $n$ the distance between $\rho_{LK}^{(n)}$ and $\SK{\emptyset} \otimes \SL{\bot}$ measured by $\delta$ is $1$. We have thus constructed a sequence that converges with respect to $d_{\mathcal{M}_{\otimes}}$, but not with respect to $\delta$. This immediately implies that these two metrics induce different topologies. We formalize this intuitive argument by the following proposition. \Cref{thm:GPT} then follows from \cref{prop:GPT_topo,lem:GPT_stability}.

One may wonder what part of the proof of \cref{prop:topo} does not generalize to an arbitrary GPT. The proof uses a nested family of $d_{\mathcal{M}}$-continuous and finite-dimensional projectors $\{\Pi^{(n)}\}_{n \in \mathbb{N}}$ such that for any state $\rho$ we have that $\lim_{n \to \infty} \delta(\Pi^{(n)} \rho \Pi^{(n)}, \rho) = 0$. In quantum theory, the basis of the Hilbert space provides such a family, while this is unclear for a general GPT.

\begin{proposition}\label{prop:GPT_topo}
  The metrics $d_{\mathcal{M}_\otimes}$ and $\delta$ do not induce the same topology on $\mathcal{S}_{KL}$.
\end{proposition}

\begin{proof}
  Consider the sequence of states $(\rho_{KL}^{(n)})_{n \in \mathbb{N}}$ given by 
  \begin{align}
    \rho_{KL}^{(n)} \coloneq  \frac{1}{2^n} \sum_{k \in \{0,1\}^n} \SK{k} \otimes \SL{k}. 
  \end{align}
  First note that, $\forall n \in \mathbb{N}: \delta(\rho^{(n)}_{KL}, \SK{\emptyset} \otimes \SL{\bot}) = 1$. To see this, consider the measurement given by $\{\EKL{n}{\yes}, \mathbf{1}_{KL}- \EKL{n}{\yes}\}$; see~\eqref{eq:jointeffect}.
  For this measurement, we find that 
  \begin{align}
    \begin{split}
      \EKL{n}{\yes}(\rho_{KL}^{(n)}) &= \frac{1}{2^n} \sum_{k,s \in \{0,1\}^{n}}  \EK{n}{s}(\SK{k}) \EL{s}{\yes}(\SL{k}) \\
      &= \frac{1}{2^n} \sum_{k,s \in \{0,1\}^{n}} \delta_{s,k} \EL{s}{\yes}(\SL{k}) = 1
    \end{split}
  \end{align}
  whereas 
  \begin{align}
    \EKL{n}{\yes}(\SK{\emptyset} \otimes \SL{\bot})  = \sum_{s \in \{0,1\}^{n}}  \EK{n}{s}(\SK{\emptyset})\EL{s}{\yes}(\SL{\bot}) = 0.
  \end{align}
  Thus, $\delta(\rho_{KL}^{(n)}, \SK{\emptyset} \otimes \SL{\bot}) = 1$.
  
  Let us now show that with respect to $d_{\mathcal{M}_\otimes}$ the sequence $\rho_{KL}^{(n)}$ converges to $\SK{\emptyset} \otimes \SL{\bot}$. Local measurements of key and lock systems are convex mixtures of measurements of the form $M_K \otimes M_L = \{\sum_{\ell \in J_i \subset \{0,1\}^n}\EK{n}{\ell} \otimes \EL{s}{r}\}_{i,r}$ where  $\cup_i J_i = \{0,1\}^n$ and $s \in \{0,1\}^*$. Therefore, it suffices to consider these measurements to calculate $d_{\mathcal{M}_\otimes}(\rho_{KL}^{(n)},\SK{\emptyset} \otimes \SL{\bot})$. Furthermore, note that when $|s| < n$, the measurement $\{\sum_{\ell \in I_j \subset \{0,1\}^n}\EK{n}{\ell} \otimes \EL{s}{r}\}_{j,r}$ is a convex mixture of measurements with $|s| = n$. Moreover, when a measurement is applied to the states $\rho_{KL}^{(n)}, \SK{\emptyset} \otimes \SL{\bot}$ and $|s| > n$, the resulting probability distribution is the same as when $s$ is replaced by the bit string $s'$ given by the first $n$ bits of $s$. Therefore, we only need to consider bit strings $s$ of length $n$.
  It is also useful to note that
  \begin{equation}\label{eq:identity}
      |\EK{n}{s} \otimes \mathbf{1}_L (\rho_{KL}^{(n)}-\SK{\emptyset} \otimes \SL{\bot})| 
      = |2^{-n} \sum_{k \in \{0,1\}^n}  \EK{n}{s}(\SK{k} - \SK{\emptyset})| = 0.
  \end{equation}
  Thus, for any such measurement $M_K \otimes M_L$
  \begin{align}
    \begin{split}
       \|P_{M_K \otimes M_L}(\rho_{KL}^{(n)}) - &P_{M_K \otimes M_L}(\SK{\emptyset} \otimes \SL{\bot})\|_1 \\
       &\leq \sum_{\ell \in  \{0,1\}^n} \sum_{r \in \{\yes, \no \}} \big| \EK{n}{\ell} \otimes \EL{s}{r} (\rho_{KL}^{(n)} - \SK{\emptyset} \otimes \SL{\bot})\big| \\
       &= 2 \sum_{\ell \in  \{0,1\}^n}  \EK{n}{\ell} \otimes \EL{s}{\yes} (\rho_{KL}^{(n)}) \\
        &\leq  2^{-n+1} \sum_{\ell \in  \{0,1\}^n, k \in \{0,1\}^n}  \EK{n}{\ell} \otimes \EL{s}{\yes} (\SK{k} \otimes \SL{k}) \\
        &=  2^{-n+1} \sum_{\ell \in  \{0,1\}^n}  \EK{n}{\ell} (\SK{s}) = 2^{-n+1}
    \end{split}
  \end{align}
  where in the first equality we used \cref{eq:identity}, $\EL{s}{\no} = \mathbf{1}_L -\EL{s}{\yes}$ and $\EL{s}{\yes}(\SK{\emptyset} \otimes \SL{\bot})=0$. Thus, $\lim_{n \to \infty} d_{\mathcal{M}_\otimes}(\rho_{KL}^{(n)},\SK{\emptyset} \otimes \SL{\bot}) = 0$. 
  
  In summary, we have found a sequence $(\rho_{KL}^{(n)})_{n \in \mathbb{N}}$ of states that converges with respect to~$d_{\mathcal{M}_\otimes}$, but not with respect to $\delta$. So, the topologies induced by these two metrics are not identical on the state space.
\end{proof}

\begin{lemma}\label{lem:GPT_stability}
  The measurement set $\mathcal{M}_\otimes$ is not fragile.
\end{lemma}
\begin{proof}
  For any effect $E \in \mathcal{E}_{SK}$ there exists an $m \in \mathbb{N}$ such that $E = \sum_{i = 1}^{m} p_i E_{K,i} \otimes E_{L,i}$. Let $(\rho_n)_{n \in \mathbb{N}}$ be a sequence that converges to $\rho$ with respect to $d_{\mathcal{M}_\otimes}$. Then this sequence also converges with respect to $d_{\mathcal{M}_\otimes \cup \{E, \mathbf{1}_{SK}-E\}}$, as 
  \begin{align}
    \begin{split}
      \lim_{n \to \infty} d_{\mathcal{M}_\otimes \cup \{E, \mathbf{1}_{SK}-E\}}(\rho_n,\rho) &\leq \lim_{n \to \infty} d_{\mathcal{M}_\otimes}(\rho_n,\rho) + \lim_{n \to \infty} |E(\rho_n-\rho)| \\
      &\leq \lim_{n \to \infty} \sum_{i = 1}^m p_i (E_{K,i} \otimes E_{L,i})(\rho_n) \\
      &\leq \sum_{i = 1}^{m} p_i  \lim_{n \to \infty} d_{\mathcal{M}_{\otimes}}(\rho_n,\rho) = 0,
    \end{split}
  \end{align}
  where we used in the last inequality that for any product effect $E_{K} \otimes E_{L}$ the measurement $M = \{E_{K} \otimes E_{L}, (\mathbf{1}_K-E_{K}) \otimes E_{L}, E_{K} \otimes (\mathbf{1}_L-E_{L}), (\mathbf{1}_K-E_{K}) \otimes (\mathbf{1}_L-E_{L})\}$ is a local measurement.
\end{proof}

\end{document}